\documentclass[11pt]{article}
\usepackage[utf8]{inputenc}
\usepackage[english]{babel}
\usepackage{graphicx}
\usepackage{amsmath}
\usepackage{amssymb}
\usepackage{amsfonts}
\usepackage{amsthm}
\usepackage[left=2.1cm,top=2cm,right=2.1cm, bottom=2cm,letterpaper]{geometry}
\usepackage{mathrsfs}
\usepackage{caption}
\usepackage{subcaption}
\usepackage{latexsym}
\usepackage{xfrac}
\usepackage{enumitem}
\usepackage{enumitem}
\usepackage{float}
\usepackage{hyperref}
\usepackage{bigints}

\newtheorem{theorem}{Theorem}[section]

\newtheorem{lemma}{Lemma}[section]
\newtheorem{definition}{Definition}[section]
\newtheorem{proposition}{Proposition}[section]

\newcommand{\N}{\mathbb{N}}
\newcommand{\R}{\mathbb{R}}

\begin{document}
\title{Thresholds for hanger slackening and cable shortening\\
in the Melan equation for suspension bridges}
\author{Filippo Gazzola - Gianmarco Sperone\\
{\small Dipartimento di Matematica, Politecnico di Milano, Italy}}
\date{}
\maketitle
\begin{abstract}
The Melan equation for suspension bridges is derived by assuming small displacements of the deck and inextensible hangers. We determine the thresholds for
the validity of the Melan equation when the hangers slacken, thereby violating the inextensibility assumption. To this end, we preliminarily study the possible
shortening of the cables: it turns out that there is a striking difference between even and odd vibrating modes since the former never shorten.
These problems are studied both on beams and plates.
\end{abstract}

\section{Introduction}
In 1888, the Austrian engineer Josef Melan \cite{melan1906theory} introduced the so-called deflection theory and applied it to derive the differential
equation governing a suspension bridge, modeled as a combination of a string (the sustaining cable) and a beam (the deck), see Figure \ref{bridgemodel}.
The beam and the string are connected through hangers. Since the spacing between hangers is usually small relative to the span, the set of the hangers is
considered as a continuous membrane connecting the cable and the deck.

\begin{figure}[H]
\begin{center}
\includegraphics[height=40mm, width=120mm]{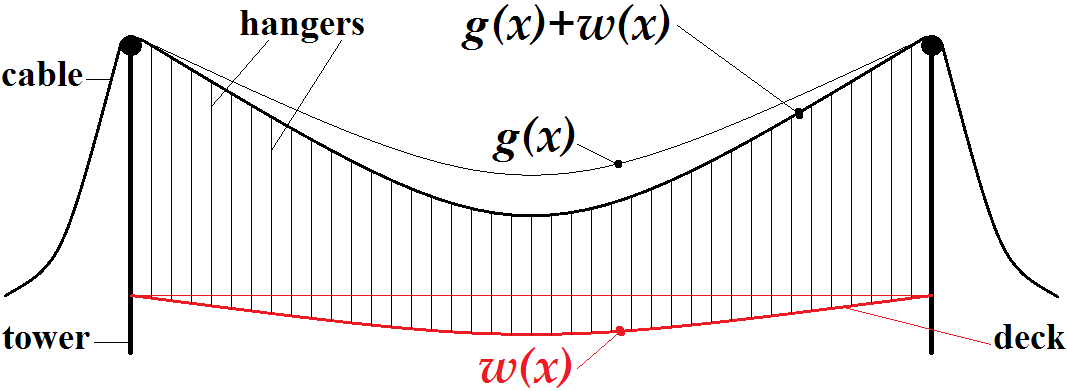}
\caption{Beam (red) sustained by a cable (black) through parallel hangers.}\label{bridgemodel}
\end{center}
\end{figure}
Let us quickly outline how the Melan equation is derived; we follow here \cite[VII.1]{von1940mathematical}.
We denote by\par\noindent
$L$ the length of the beam at rest (the distance between towers) and $x\in(0,L)$ the position on the beam;\par\noindent
$p=p(x)$ the live load and $-q<0$ the dead load per unit length applied to the beam;\par\noindent
$g=g(x)$ the displacement of the cable due to the dead load $-q$;\par\noindent
$L_c$ the length of the cable subject to the dead load $-q$;\par\noindent
$A$ the cross-sectional area of the cable and $E_c$ its modulus of elasticity;\par\noindent
$H$ the horizontal tension in the cable, when subject to the dead load $-q$ only;\par\noindent
$EI$ the flexural rigidity of the beam;\par\noindent
$w=w(x)$ the displacement of the beam due to the live load $p$;\par\noindent
$h=h(w)$ the additional tension in the cable produced by the live load $p$.\par\smallskip

When the system is only subject to the action of dead loads, the cable is in position $g(x)$ while the unloaded beam is in the horizontal position $w\equiv0$,
see Figure \ref{bridgemodel}. The cable is adjusted in such a way that it carries its own weight, the weight of the hangers and the weight of the
deck (beam) without producing a bending moment in the beam, so that all additional deformations of the cable and the beam due to live loads are small.
The cable is considered as a perfectly flexible string subject to vertical dead and live loads. The string is subject to a downwards vertical constant
dead load $-q$ and the horizontal component $H>0$ of the tension remains constant. If the mass of the cable is neglected, then the dead load is
distributed per horizontal unit. The resulting equation simply reads $Hg''(x)=q$ (see \cite[(1.3),VII]{von1940mathematical})
so that the cable takes the shape of a parabola with a $\cup$-shaped graph.
If the endpoints of the string (top of the towers) are at the same level $\gamma>0$ (as in suspension bridges, see again Figure \ref{bridgemodel}),
then the solution $g$ and the length $L_c$ of the cable are given by:
\begin{equation}\label{trueparabola}
g(x)\!=\!\gamma\!+\!\frac{q}{2H}x(x-L)\, ,\quad g'(x)=\frac{q}{H}\left(x-\frac{L}{2}\right)\, ,\quad g''(x)=\frac{q}{H}\, ,\quad\forall x\in(0,L),
\end{equation}
\begin{equation} \label{trueparabola2}
L_c\!=\!\int\limits_0^{L}\!\sqrt{1\!+\!g'(x)^2}\, dx.
\end{equation}

The elastic deformation of the hangers is usually neglected, so that the function $w$ describes both the displacements of the beam and of the cable from its
equilibrium position $g$. This classical assumption is justified by precise studies on linearized models, see e.g.\ \cite{luco}.
When the live load $p$ is added, a certain amount $p_1$ of $p$ is carried by the cable whereas the remaining part $p-p_1$ is carried by the bending
stiffness of the beam. In this case, it is well-known \cite{melan1906theory,von1940mathematical} that the equation for the displacement $w$
of the beam is
\begin{equation}\label{terza}
EI\, w''''(x)=p(x)-p_1(x)\qquad\forall x\in(0,L)\, .
\end{equation}
At the same time, the horizontal tension of the cable is increased to $H+h(w)$ and the deflection $w$ is added to the displacement $g$. Hence,
according to \eqref{trueparabola}, the equation which takes into account these conditions reads
\begin{equation}\label{seconda}
\big(H+h(w)\big)\big(g''(x)+w''(x)\big)=q-p_1(x)\qquad\forall x\in(0,L)\, .
\end{equation}
Then, by combining \eqref{trueparabola}-\eqref{terza}-\eqref{seconda}, we obtain
\begin{equation}\label{melaneq}
EI\, w''''(x)-\big(H+h(w)\big)\, w''(x)-\frac{q}{H}\, h(w)=p(x)\qquad\forall x\in(0,L)\, ,
\end{equation}
which is known in literature as the {\bf Melan equation} \cite[p.77]{melan1906theory}. The beam representing the bridge is hinged at its
endpoints, which means that the boundary conditions to be associated to \eqref{melaneq} are
\begin{equation}\label{hinged}
w(0)=w(L)=w''(0)=w''(L)=0\, .
\end{equation}

Theoretical results on the Melan equation \eqref{melaneq} are quite demanding \cite{gazzola2014melan,gazzola2016variational} and this is the reason why
it has attracted the attention of numerical analysts \cite{se1,se,se2,woll}.
In this paper we analyze and quantify the two main nonlinear (and challenging) behaviors of \eqref{melaneq}.
The first one is the additional tension of the cable, $h(w)$ which is a nonlocal term and is proportional to the length increment of the cable.
Depending on the deflection of the beam, the cable may vary its shape and tension, and such phenomenon is studied in Section \ref{thresholdcable}
where we compute the exact thresholds of shortening, depending on the deflection $w$. In Theorem \ref{theoshortening} we show that there is a
striking difference between the even and odd vibrating modes of the beam. The second source of nonlinearity is the possible slackening
of the hangers which, however, is not considered in \eqref{melaneq} due to the assumption of inextensibility of the hangers.
Indeed, $w$ in \eqref{melaneq} aims to represent both the deflections of the beam and of the cable, implying that the cable reaches the new
position $g+w$. But since the hangers do not resist to compression, they may slacken so that
the cable and the beam move independently and $w$ will no longer represent the displacement of the cable from its original position. This phenomenon
is analyzed in detail in Section \ref{hangers} where we suggest an improved version of \eqref{melaneq} which also takes into account the slackening of
the hangers, see \eqref{improved}. In Section \ref{sectplate} we extend this study to a partially hinged rectangular plate aiming to model the deck of a
bridge and thereby having two opposite edges completely free: we view these free edges as beams sustained by cables and
governed by the Melan equation. The results are complemented with some enlightening figures.

\section{Thresholds for cable shortening in a beam model} \label{thresholdcable}

A given displacement of the deck $w\in C^1([0, L], \mathbb{R})$ generates an additional tension $h(w)$ in the cable that is proportional to the increment of
length of the cable $\Gamma(w)$, that is,
\begin{equation} \label{gamma}
h(w) = \frac{E_cA}{L_{c}}\Gamma(w)\ \mbox{ where }\ \Gamma(w) = \int\limits_{0}^{L}\Big[\sqrt{1 + \big(w'(x)+g'(x)\big)^2}-\sqrt{1 + g'(x)^2}\Big] \ dx\, .
\end{equation}	

\begin{definition}
We say that a displacement $w$ \textbf{shortens} the cable if $\Gamma(w)<0$.
\end{definition}

There are at least three rude ways to approximate $h(w)$, by replacing $\Gamma(w)$ with
$$
-\tfrac{q}{H}\!\int_0^L\!\!w(x)dx,\quad-\tfrac{q}{H}\!\int_0^L\!\!w(x)dx+\int_0^L\!\tfrac{w'(x)^2}{2}dx,\quad
-\tfrac{q}{H}\!\!\bigintsss_0^L\!\tfrac{w(x)}{\left[1+\tfrac{q^2}{H^2}\left(x-\tfrac{L}{2}\right)^2\right]^{3/2}}dx.
$$
These approximations are obtained through an erroneous argument. While introducing \eqref{melaneq}, Biot-von K\'arm\'an \cite{von1940mathematical}
warn the reader by writing {\em whereas the deflection of the beam may be considered small, the deflection of the string, i.e., the deviation of its shape
from a straight line, has to be considered as of finite magnitude}. However, they later decide to {\em neglect $g'(x)^2$ in comparison with unity}.
A similar mistake with a different result is repeated by Timoshenko \cite{timo1,timo2}. These approximations may lead to an average error of about $5\%$ for $h(w)$. Around $1950$ the civil and structural German engineer Franz Dischinger emphasized the dramatic consequences of bad approximations on the structures and $5\%$ turns out to be a too large error.
Moreover, since related numerical procedures are very unstable, see \cite{gazzola2014melan,se1,se,se2}, also from a mathematical point of view one
should analyze the term $h(w)$ with extreme care.\par
Since the displacement of the deck $w$, created by a live load $p$, is the solution of the Melan equation \eqref{melaneq}, we study here
which loads yield a shortening of the cable. In particular, we analyze the fundamental modes of vibration of the beam so that we consider
the following class of live loads:
\begin{equation} \label{live_load1}
p_n(x) =\rho\left(\frac{n \pi}{L} \right)^{2} \left\{ \left(\frac{n \pi}{L} \right)^{2} EI + H + h\left(\rho\sin \left(\tfrac{n \pi x}{L} \right)\right)\right\}
\sin \left(\frac{n \pi x}{L} \right) - \frac{q}{H}h \left(\rho\sin \left(\tfrac{n \pi x}{L} \right)\right)\quad\forall n\in \mathbb{N},
\end{equation}
for varying values of $\rho \in \mathbb{R}$. The load $p_n$ consists of a negative constant part $-\frac{q}{H}h\big(\rho\sin(\frac{n \pi x}{L})\big)$ and a part
that is proportional to the fundamental vibrating modes of the beam $\sin \left(\frac{n \pi x}{L} \right)$, which are the eigenfunctions of the
following eigenvalue problem:
\begin{equation}\label{eigenbeam}
v''''(x)=\lambda v(x)\quad(0<x<L)\, ,\qquad v(0)=v(L)=v''(0)=v''(L)=0\, .
\end{equation}
The reason of this choice for $p_n$ is that, after some computations, one sees that the resulting displacement $w_n$ (solution of \eqref{melaneq}) is proportional
to a vibrating mode:
\begin{equation} \label{melan}
w_{n}(x) = \rho\sin \left(\frac{n \pi x}{L} \right)\quad\forall x \in [0,L].
\end{equation}
Whence, $|\rho|$ measures the amplitude of oscillation of the vibrating mode $w_{n}$. For every $n \in \mathbb{N}$, we put $\Gamma_{n}(\rho):=\Gamma(w_n)$
and from \eqref{gamma} we infer that

\begin{equation} \label{gamma2}
\Gamma_{n}(\rho) = \int\limits_{0}^{L} \sqrt{1 + \left[\dfrac{q}{H}\left(x - \dfrac{L}{2} \right)+\dfrac{n\pi}{L}\rho\cos
\left(\frac{n \pi x}{L} \right)  \right] ^2} \ dx - L_{c}\quad\forall\rho\in\mathbb{R}.
\end{equation}

In the next result we emphasize a striking difference between even and odd modes.

\begin{theorem}\label{theoshortening}
Assume that $\tfrac{q}{H}<\tfrac{2}{5}$.\par\noindent
$\bullet$ If $n \geq 1$ is even, then $\Gamma_{n}(\rho)\ge0$ for all $\rho$; therefore, an even vibrating mode cannot shorten the cable.\par\noindent
$\bullet$ If $n \geq 1$ is odd, then there exists a (unique) critical value $\rho_{n}^{*}>0$ such that $\Gamma_{n}(\rho_{n}^{*})=0$ and $\Gamma_{n}(\rho)<0$
for all $\rho\in(0,\rho_{n}^{*})$; therefore, odd vibrating modes shorten the cable when their amplitude of oscillation $\rho$ is within this interval.
\end{theorem}

Theorem \ref{theoshortening} is proved in Section \ref{proof1}. The assumption $q/H<2/5$ in Theorem \ref{theoshortening} is verified in the vast
majority of real suspension bridges. For instance, for the numerical data employed in \cite{woll}, it happens that
$q/H=1.739 \times 10^{-3} \, [m^{-1}]$. Moreover, as reported in \cite[Section 15.17]{podolny}, the sag-span ratio in a suspension bridge always lies
in the range $(\tfrac{1}{12},\tfrac{1}{8})$. In view of \eqref{trueparabola}, this means that
$$
\dfrac{L}{12} <  g(0) - g \left( \dfrac{L}{2} \right) < \dfrac{L}{8}\quad\mbox{or, equivalently,}\quad \dfrac{2}{3L} < \dfrac{q}{H} < \dfrac{1}{L}.
$$
Therefore, the assumption $\tfrac{q}{H}<\tfrac{2}{5}$ is valid for any suspension bridges with a span of at least $2.5\, [m]$! In any case, numerical
results seem to show that the assumption $\tfrac{q}{H}<\tfrac{2}{5}$ is not necessary for the validity of Theorem \ref{theoshortening}.\par
Related to $\rho_n^*$, as characterized by Theorem \ref{theoshortening}, we introduce the quantity
\begin{equation} \label{xi}
\xi_{n}^*=\rho_{n}^*\left(\frac{n \pi}{L} \right)^{2} \left\{ \left(\frac{n \pi}{L} \right)^{2} EI + H + \frac{E_cA}{L_{c}} \Gamma\big(\rho_n^*
\sin \left(\tfrac{n \pi x}{L}\right)\big)\right\}\quad\forall n \in \mathbb{N},
\end{equation}
which is the amplitude of oscillation of the live load $p_n$ in \eqref{live_load1} that generates the critical oscillation $w_{n}^*(x)=\rho_n^*
\sin(\tfrac{n \pi x}{L})$. Throughout this paper, as far as numerical data are needed, we use the parameters taken from \cite{woll}:
\begin{equation}\label{numerical}
L = 460\, [m],\quad EI = 57 \times 10^{6}\, [kN \cdot m],\quad E_cA=36\times 10^{6} \, [kN],\quad\frac{q}{H} = 1.739 \times 10^{-3} \, [m^{-1}].
\end{equation}
Table \ref{table:1} shows the critical values of $\rho_{n}^{*}$ and $\xi_{n}^{*}$ (according to Theorem \ref{theoshortening}
and \eqref{xi}), as functions of some odd values of $n \in \N$.
\begin{table}[H]
\small
   \centering
   \begin{tabular}{ | c | c | c | c | c | c | c | c | c | c | c | c | c |}
       \hline
       $n$ & 1 & 3 & 5 & 7 & 9 & 11 & 13 & 15 & 17 & 19 \\ \hline
       $\rho_{n}^{*}$ & 94.807 & 3.056 & 0.657 & 0.239  & 0.112 & 0.061 & 0.037 & 0.024 & 0.016 & 0.011 \\ \hline
       $\xi_{n}^{*}$ & 444.016 & 156.115 & 125.811 & 124.578 & 132.962 & 145.676 & 160.734 & 177.192 & 194.559 & 212.620 \\ \hline
   \end{tabular}
   \captionsetup{justification=centering}
   \captionsetup{font={footnotesize}}
   \captionof{table}{Critical coefficients for cable shortening in odd-vibrating modes.}
   \label{table:1}
\end{table}

As stated in Theorem \ref{theoshortening}, even modes never shorten the cable. This {\em does not} mean that odd modes are ``worse'' or more prone to elongate
the cable. On the contrary, thinking of a periodic-in-time oscillation proportional to a vibrating mode \eqref{melan}, that is,
$$
\rho(t)\sin \left(\frac{n \pi x}{L} \right)\quad\forall x \in [0,L],\quad\forall t>0\, ,
$$
with $\rho(t)$ varying between $\pm\overline{\rho}$, we reach the opposite conclusion.
To see this, in Figure \ref{subfig:xi1} we plot the graphs of $\Gamma_{2}$ and $\Gamma_{3}$ and we see that
$$\max\{\Gamma_3(\overline{\rho}),\Gamma_3(-\overline{\rho})\}>\max\{\Gamma_2(\overline{\rho}),\Gamma_2(-\overline{\rho})\}=\Gamma_2(\overline{\rho})\, .$$
Therefore, even if the cable shortens when $\rho(t)\in(0,\rho_3^*)$ for the third mode, the cable itself elongates more than for the second mode
when $\rho(t)<0$. We come back to this issue in Section \ref{sectplate}.

\begin{figure}[H]
\centering
\begin{subfigure}{.44\textwidth}
  \includegraphics[scale=0.5]{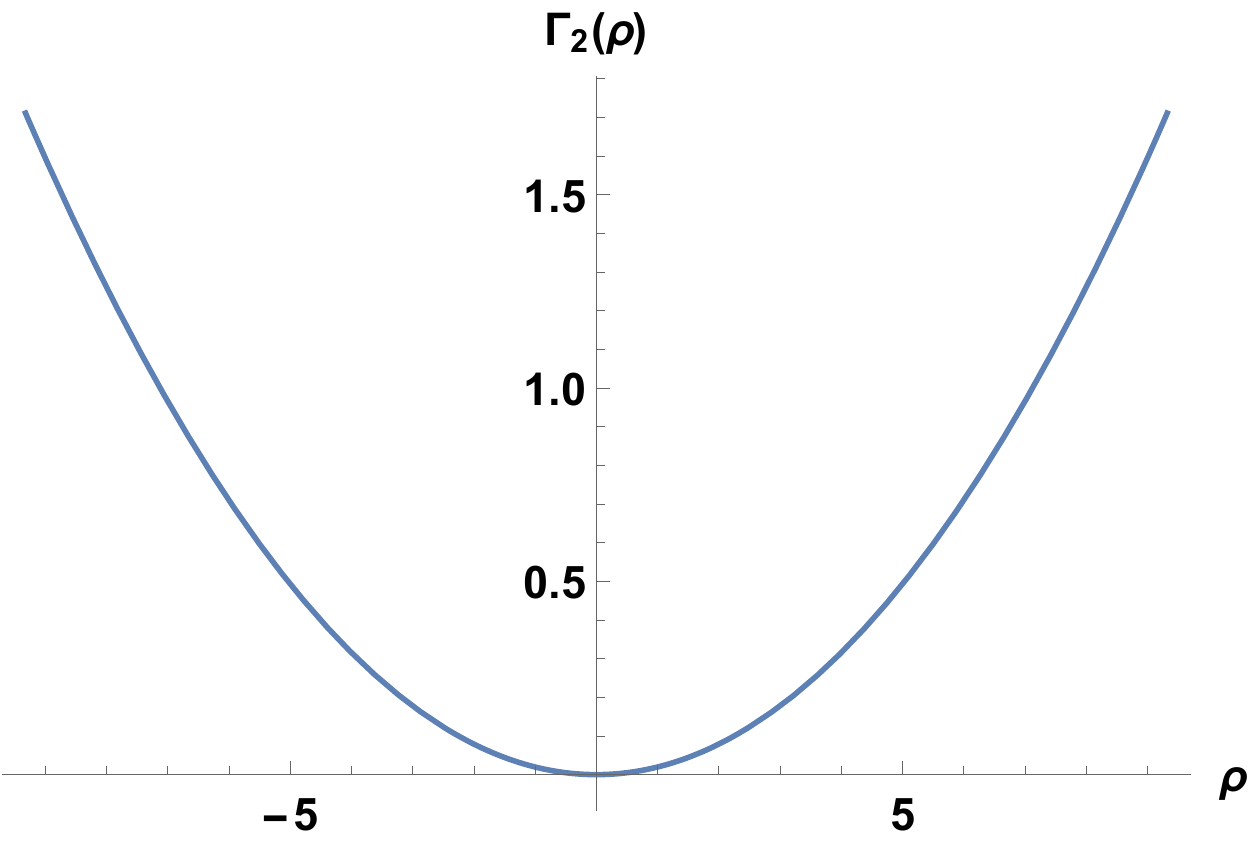}
\end{subfigure}
\quad
\begin{subfigure}{.44\textwidth}
   \includegraphics[scale=0.5] {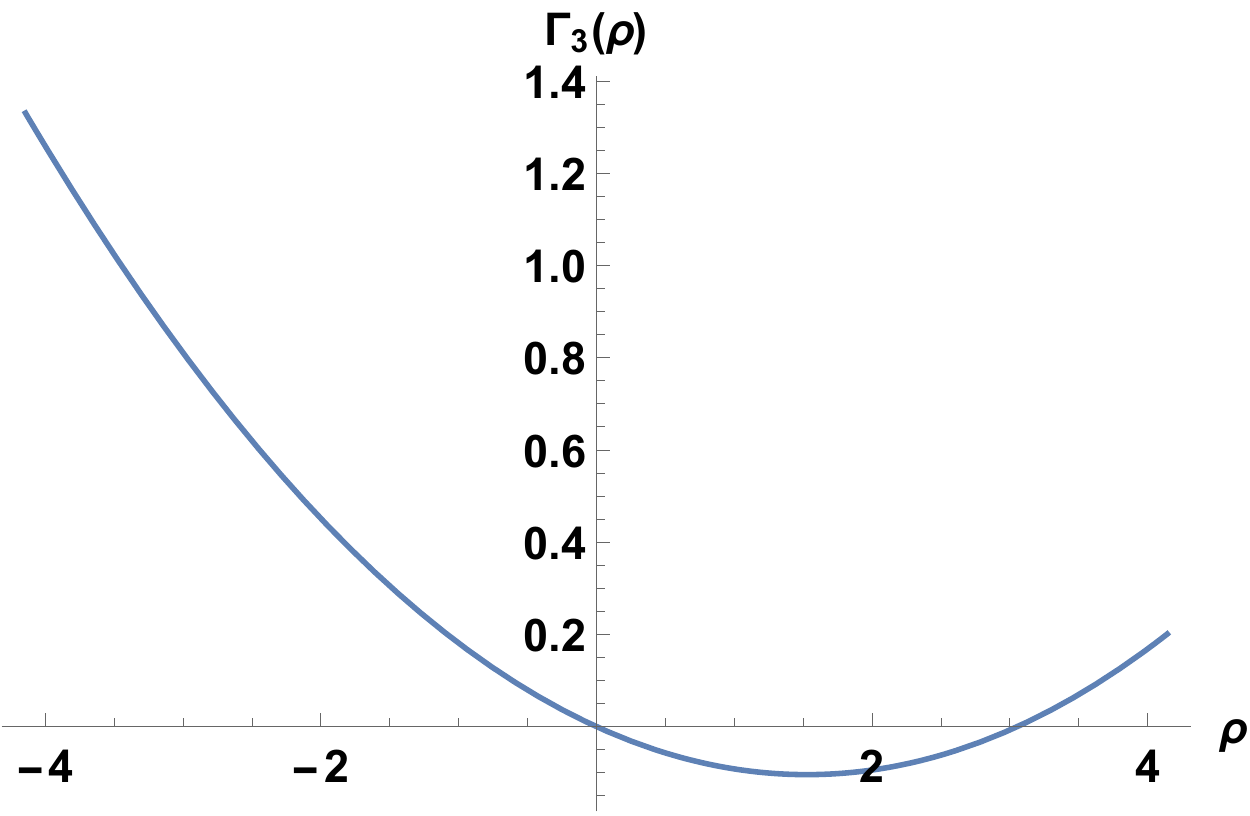}
\end{subfigure}
\caption{Increment of cable length in the second (left) and third (right) vibrating modes.}\label{subfig:xi1}
\end{figure}

\section{Thresholds for hangers slackening in a beam model}\label{hangers}

In this section we estimate the thresholds that provoke the slackening of some hangers. Since the hangers resist to extension but not to compression,
if the deck goes too high above its equilibrium position, then {\em the hangers may no longer be considered as rigid inextensible bars}. In particular,
they will not push upwards the cable in such a way that it loses convexity: the general principles governing the deformation of a finite-length
cable under the action of a downwards vertical load (see \cite[(1.3), VII]{von1940mathematical}) indicate that the cable remains convex.
This means that if $g+w$ is not convex, then it {\em does not} describe the position of the cable anymore.\par
In order to explain how the Melan equation \eqref{melaneq} should be modified in case of hanger slackening we briefly recall the concept of
\textit{convexification} which can be formalized in several equivalent ways, see \cite[(3.2), I]{ekeland1976convex} for full details.\par
Let $I\subset\R$ be a compact interval. The convexification $f^{**}$ of a continuous function $f:I\to\R$ is:\par
$\bullet$ the pointwise supremum of all the affine functions everywhere less than $f$;\par
$\bullet$ the pointwise supremum of all the convex functions everywhere less than $f$;\par
$\bullet$ the largest convex function everywhere less than or equal to $f$;\par
$\bullet$ the convex function whose epigraph is the closed convex hull of the epigraph of $f$;\par
$\bullet$ the second Fenchel conjugate of $f$, that is,
$$f^{**}(x)=\sup_{y\in\R} \{ xy - f^{*}(y) \}\ \ \forall x\in I\, ,\quad\mbox{ where }\quad f^{*}(y)=\max_{x\in I}\{yx-f(x)\}\ \ \forall y\in\R\, .$$

This notion enables us to give the following:

\begin{definition} \label{slacken_region0}
We say that a displacement $w$ \textbf{slackens} the hangers in some (nonempty) interval $(a,b)\subset[0,L]$ if the graph of
\begin{equation}\label{z}
z:=g+w
\end{equation}
lies strictly above that of its convexification $z^{**}$ in $(a,b)$. Then, the \textbf{slackening region} $\mathcal{S}\subset[0,L]$ is the union
of all the slackening intervals, that is,
$$
\mathcal{S}=\{x \in (0,L) \ | \ z(x) > z^{**}(x) \}\, .
$$
\end{definition}

In the slackening region, not only the Melan equation \eqref{melaneq} is incorrect but also \eqref{terza} fails since {\em the whole amount of live
load is carried by the beam}: one has $p_1(x)=0$ for all $x\in\mathcal{S}$. Therefore \eqref{melaneq} should be replaced with the more reliable equation
\begin{equation}\label{improved}
EI\, w''''(x)+\Big(\chi_\mathcal{S}(w) - 1\Big)\bigg(\big(H+h(w)\big)\, w''(x)+\frac{q}{H}\, h(w)\bigg)=p(x)\qquad\forall x\in(0,L)
\end{equation}
where $\chi_\mathcal{S}(w)$ is the characteristic function (that depends on $w$) of the slackening region $\mathcal{S}$, see Definition
\ref{slacken_region0}. We summarize these results in the following statement.

\begin{proposition}\label{summary}
In absence of slackening ($\mathcal{S}=\emptyset$) the two equations \eqref{melaneq} and \eqref{improved} coincide; in this case, the solution
$w$ represents the displacement of the beam whereas $z$ in \eqref{z} represents the position of the cable.\par
In presence of slackening ($\mathcal{S}\neq\emptyset$) the correct equation is \eqref{improved} and the position of the cable is described by $z^{**}$.
\end{proposition}

The term $(\chi_\mathcal{S}(w)-1)$ adds a further nonlinearity to the Melan equation \eqref{melaneq}.
As far as we are aware, there is no general theory to tackle equations such as \eqref{improved}. It would therefore be interesting to study its features in detail.\par
Although the exact slackening region is difficult to determine, it is clear that the non-convexity intervals of $z$ in \eqref{z} represent
proper subsets of these regions. Therefore, we have

\begin{proposition}\label{slackening1}
Let $w$ be the solution of \eqref{improved} and let $z$ be as in \eqref{z}. If $\mathcal{S}\neq\emptyset$, then
$$\{x\in(0,L);\, z''(x)\le0\}\subsetneqq\mathcal{S}\, .$$
\end{proposition}

We now apply Proposition \ref{summary} to the case of the loads $p_n$ in \eqref{live_load1}.

\begin{proposition}\label{propmodes}
Let $p_n$ and $w_n$ be as in \eqref{live_load1} and \eqref{melan}. Let
\begin{equation} \label{umbral0}
C_{n}^{*} := \frac{q}{H} \left(\frac{L}{n \pi} \right)^{2}\quad\forall n \in \N\, .
\end{equation}
Slackening occurs if and only if
\begin{equation}\label{iff}
\rho>C_1^*\mbox{ when }n=1\, ,\qquad|\rho|>C_n^*\mbox{ when }n\ge2\, ;
\end{equation}
in this case, the position of the cable is described by $z_n^{**}$ (with $z_n=g+w_n$).
\end{proposition}

The proof of Proposition \ref{propmodes} is fairly simple. The slackening region of $w_n$ is nonempty if and only if there exists $x\in(0,L)$
such that $z_n''(x)<0$, where
$$z_{n}(x) =g(x)+w_n(x)= \gamma + \frac{q}{2H}x(x-L) + \rho\sin \left(\frac{n \pi x}{L} \right)\quad\forall x \in [0,L].$$
This property translates into
$$\exists x\in(0,L)\quad\mbox{such that}\quad\rho\sin \left(\frac{n \pi x}{L} \right)>C_n^*\, ,$$
which is equivalent to \eqref{iff}.
\begin{figure}[H]
	\centering
	\begin{subfigure}{.4\textwidth}
		\includegraphics[scale=0.5]{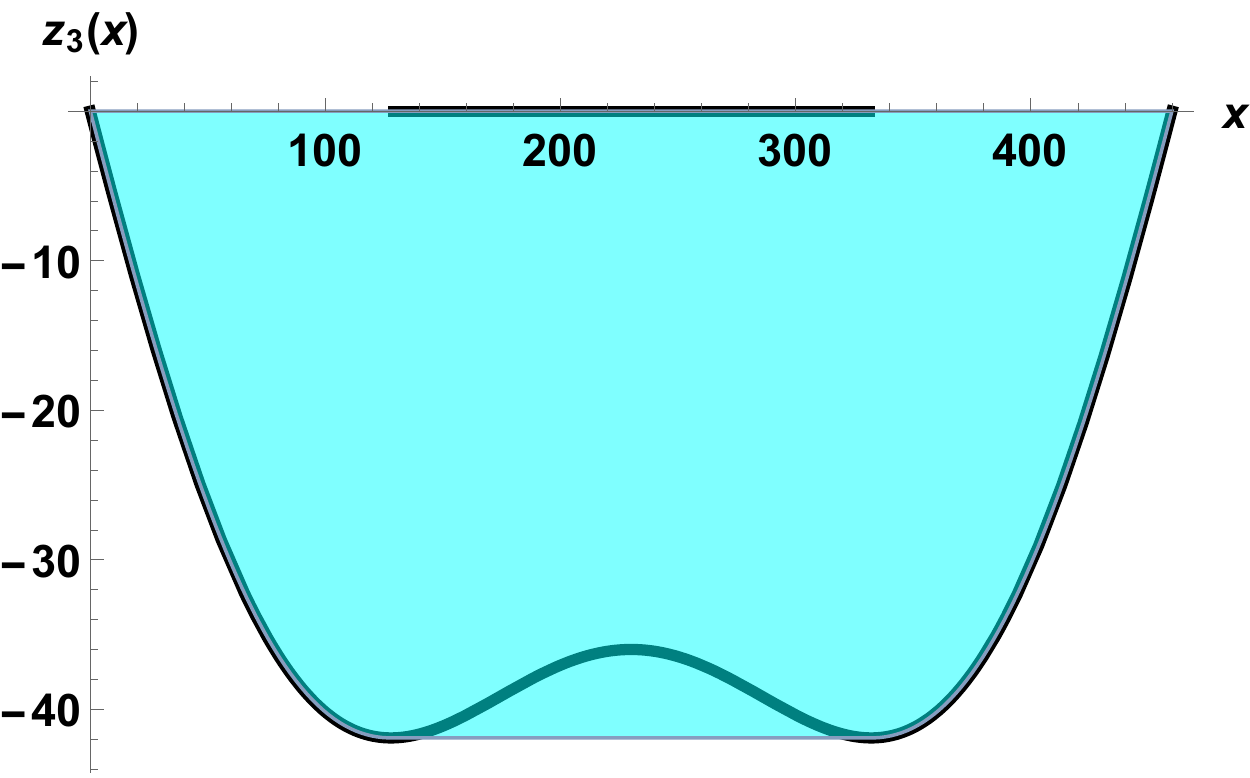}
	\end{subfigure}
	\quad
	\begin{subfigure}{.4\textwidth}
		\includegraphics[scale=0.5]{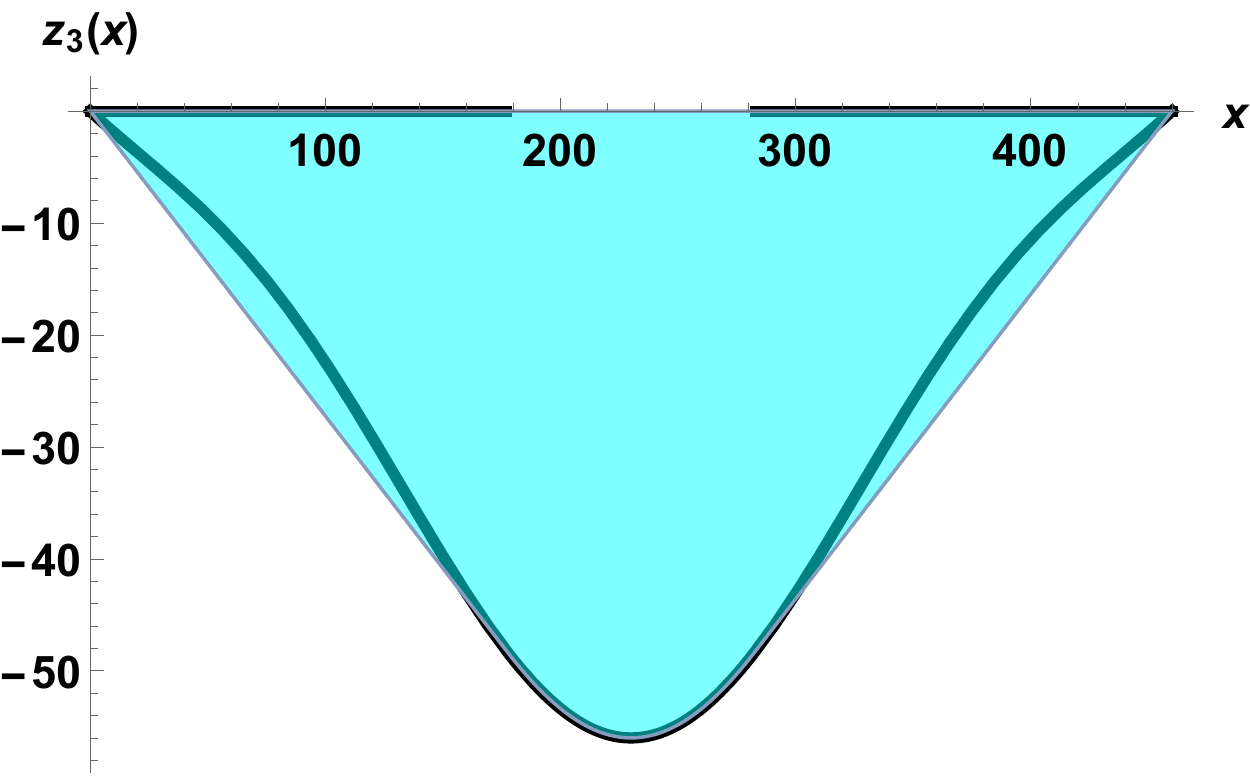}
	\end{subfigure}
	\caption{Slackening of the third vibrating mode when $\rho_{3} < -C_{3}^{*}$ (left) and when $\rho_{3} > C_{3}^{*}$ (right).}\label{fig:revol0}
\end{figure}
Since it is by far nontrivial to determine explicitly the convexification of $z_n$ and the slackening region $\mathcal{S}_{n}$, we follow a
numerical-geometrical approach, that is, we plot the closed convex hull of the epigraph of $z_n$. We take again the numerical values \eqref{numerical}.
In order to illustrate the procedure, consider the function $z_{3}$ (with $\gamma = 0$, since we are only interested in the shape of the curve), whose
slackening threshold is $C_{3}^{*} \approx 4.1426$. By putting amplitudes of $\rho_{3} = \pm 10$, we obtained the graphs of $z_{3}$ in Figure
\ref{fig:revol0} where the slackening intervals have been highlighted over the horizontal axis, and the closed convex hull of the epigraph of $z_{3}$
has been shaded.
\noindent
Similarly, by putting amplitudes of $\rho_{5} = \pm 5$, we obtained the plots displayed in Figure \ref{fig:revol1} for the graphs of $z_{5}$
(for which $C_{5}^{*} \approx 1$.$4913$):

\begin{figure}[H]
	\centering
	\begin{subfigure}{.4\textwidth}
		\includegraphics[scale=0.5]{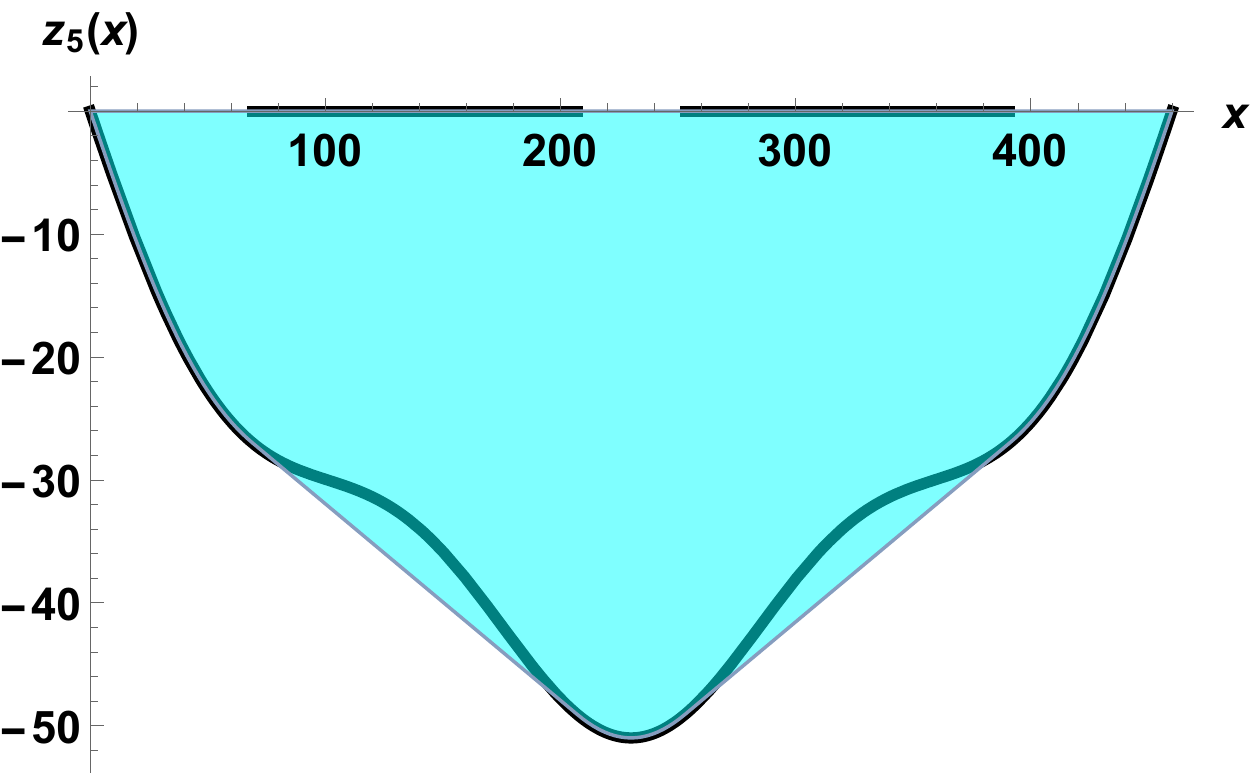}
	\end{subfigure}
	\quad
	\begin{subfigure}{.4\textwidth}
		\includegraphics[scale=0.5]{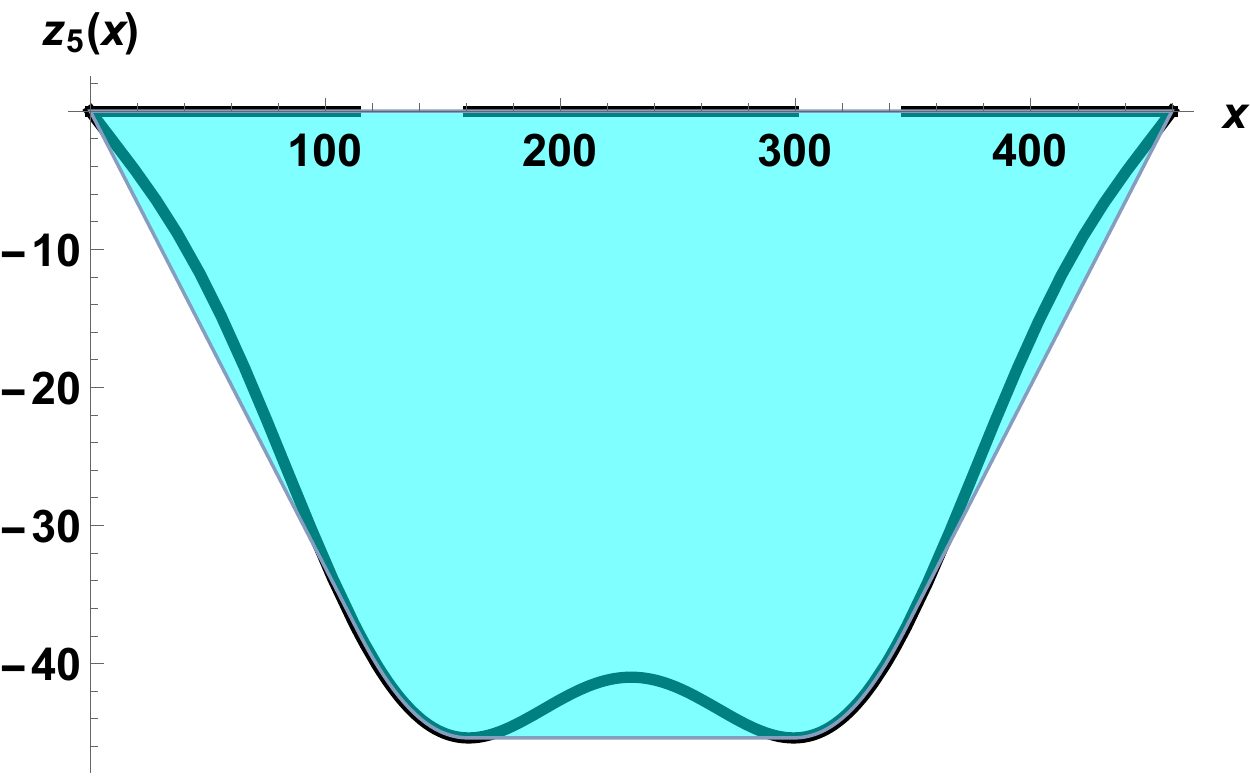}
	\end{subfigure}
	\caption{Slackening of the fifth vibrating mode when $\rho_{5} < -C_{5}^{*}$ (left) and when $\rho_{5} > C_{5}^{*}$ (right).}\label{fig:revol1}
\end{figure}

It is worthwhile noticing that the hangers slackening in even modes occurs asymmetrically with respect to the center of the beam
but, at the same time, symmetrically with respect to the value of $\rho_{n}$. To clarify this point, in Figure \ref{fig:revol2} we display the graphs
of $z_{2}$ (where $C_{2}^{*}\approx 9.3208$) when $\rho_{2} = -20$, and of $z_{4}$ (where $C_{4}^{*}\approx 2.3301$) when $\rho_{4} = 8$.
The remaining figures when $\rho_{2} > C_{2}^{*}$ or $\rho_{4} < -C_{4}^{*}$ may be obtained by simply reflecting the curves with respect to the center
of the beam.

\begin{figure}[H]
\centering
\begin{subfigure}{.4\textwidth}
    \includegraphics[scale=0.5]{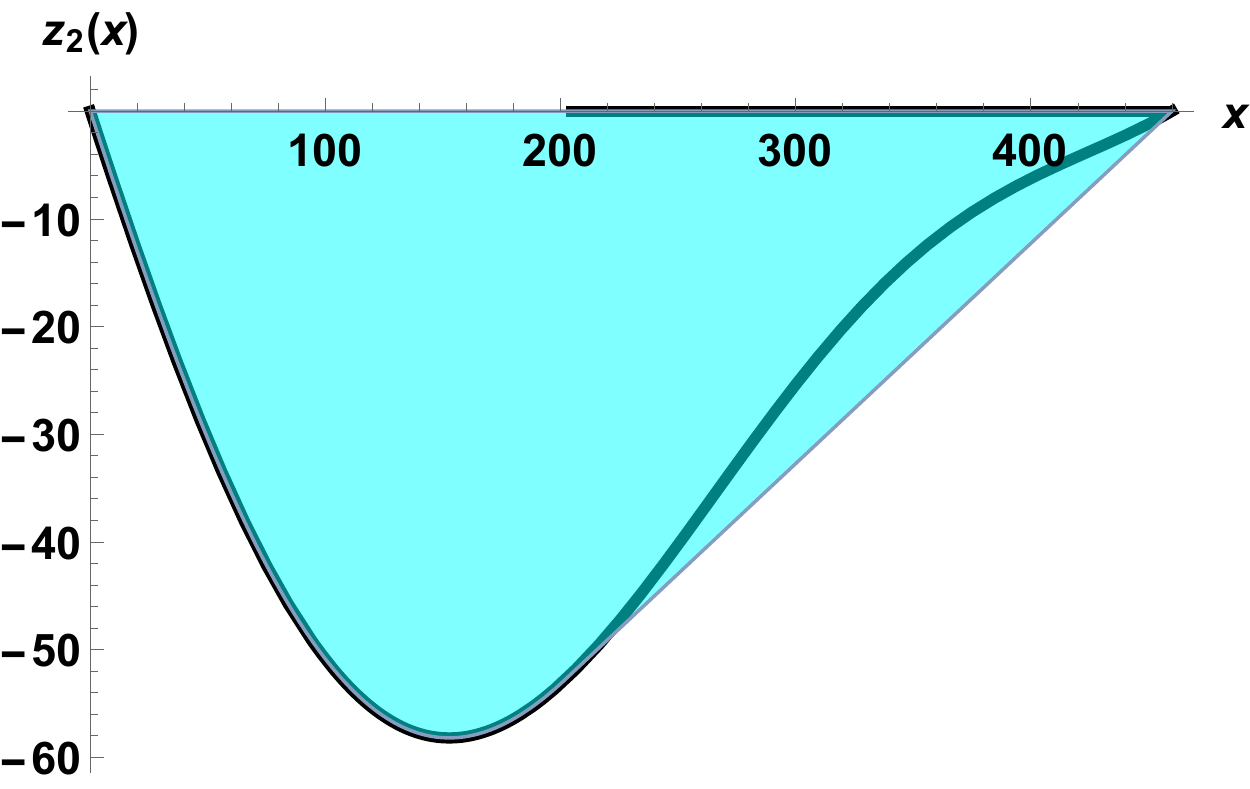}
\end{subfigure}
\quad
\begin{subfigure}{.4\textwidth}
    \includegraphics[scale=0.5]{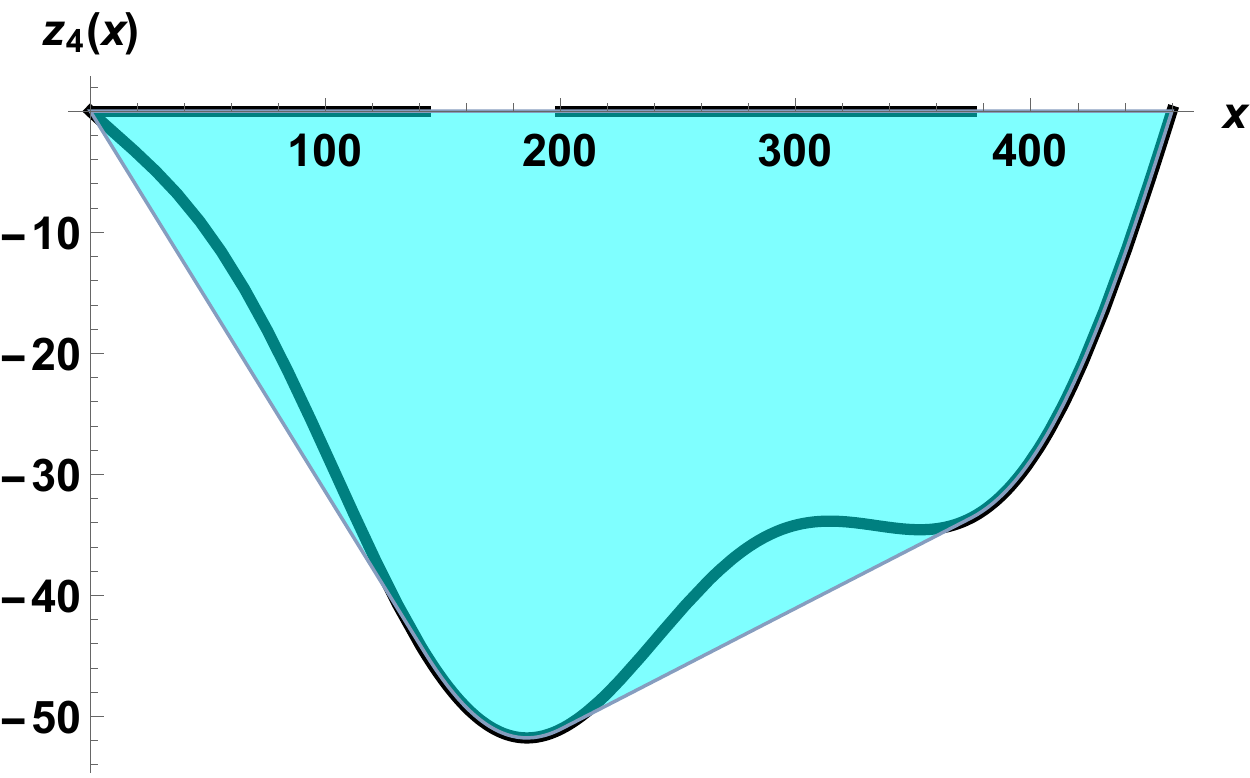}
\end{subfigure}
\caption{Slackening of the second vibrating mode when $\rho_{2} <- C_{2}^{*}$ (left), and of the fourth vibrating mode when $\rho_{4} > C_{4}^{*}$
(right).}\label{fig:revol2}
\end{figure}

The numerical values of $C_n^*$ for $n\le10$ are reported in Table \ref{table:2} where we used the parameters as in \eqref{numerical}.\par
One last issue must be addressed. In some of the pictures in Figures \ref{fig:revol0}, \ref{fig:revol1} and \ref{fig:revol2} we observe that the
endpoints of the deck $x=0$ and $x=L$ actually belong to the slackening region $\mathcal{S}_{n}$. This is clearly a physically impossible situation
since the hangers are not expected to slacken at the endpoints of the beam. Geometrically, one expects instead that the tangent lines to the
curve at the endpoints of the beam lie strictly below the graph of $z_{n}$ in $(0,L)$, that is:
\begin{equation} \label{umbral1}
z_{n}(x) > \max \{z_{n}'(0)x, \ z_{n}'(L)(x-L) \} \ \ \forall x \in (0,L), \ \ \forall n \in \N.
\end{equation}
Clearly, condition \eqref{umbral1} is not satisfied for large values of $|\rho_{n}|$, but it remains valid even when $|\rho_{n}|$ is slightly larger
than the slackening (and convexity) threshold \eqref{umbral0}. For the first ten vibrating modes,
we numerically computed the threshold $\rho_{n}^{**}$ that ensures condition \eqref{umbral1}, when $|\rho_n|\le\rho_{n}^{**}$ (if $n$ is even) and
$\rho_{n} \leq \rho_{n}^{**}$ (if $n$ is odd), with the parameters as in \eqref{numerical}. We obtained the second line in Table \ref{table:2}.

\begin{table}[H]
\small
\centering
\begin{tabular}{ | c | c | c | c | c | c | c | c | c | c | c | c | c | c |}
\hline
$n$ & 1 & 2 & 3 & 4 & 5 & 6 & 7 & 8 & 9 & 10   \\ \hline
$C_{n}^{*}$ & 37.283 & 9.321 & 4.143 & 2.330  & 1.491 & 1.035 & 0.761 & 0.582 & 0.460 & 0.372 \\ \hline
$\rho_{n}^{**}$ & 58.564 & 14.641 & 6.507 & 3.660  & 2.342 & 1.626 & 1.195 & 0.915 & 0.723 & 0.585  \\ \hline
\end{tabular}
\captionsetup{justification=centering}
\captionsetup{font={footnotesize}}
\captionof{table}{Thresholds for non-convexity and hangers slackening in the first ten vibrating modes.}
\label{table:2}
\end{table}

\section{Behavior of cables and hangers in a plate model}\label{sectplate}

The deck of a real bridge cannot be described by a simple (one-dimensional) beam since it fails to display torsional oscillations. In this section we take
advantage of the results so far obtained in order to analyze the vibrating modes of a rectangular plate $\Omega=(0,\pi)\times(-\ell,\ell)$ ($2\ell>0$ is the width
of the plate and $2\ell\ll\pi$); for simplicity, we take here $L=\pi$. Specifically, we consider a partially hinged plate whose elastic energy is given by the
Kirchhoff-Love functional, see \cite{nazarov2012hinged,sweers2009survey} for discussions on the boundary conditions and updated derivation of the corresponding
Euler-Lagrange equation. From \cite{ferrero2015partially} we know that the vibrating modes of the plate $\Omega$ are obtained by solving
the following eigenvalue problem
\begin{equation}\label{plate}
\left\lbrace
\begin{array}{ll}
\Delta^2 u = \lambda u\quad & \text{for }(x,y)\in\Omega\\
u=u_{xx}=0\quad & \text{for }(x,y)\in\{0,\pi\}\times(-\ell, \ell)\\
u_{yy}+\sigma u_{xx}=u_{yyy}+(2-\sigma)u_{xxy}=0\quad & \text{for }(x,y)\in(0,\pi)\times\{-\ell,\ell\}\, ,
\end{array}\right.	
\end{equation}
where $\sigma \in \left(0 , \frac{1}{2} \right)$ is the Poisson ratio. The boundary conditions for $x=0$ and $x=\pi$ show that the short edges of the
plate are hinged, while the conditions for $y=\pm\ell$ show that the plate is free on the long edges. Problem \eqref{plate} is the two-dimensional counterpart
of \eqref{eigenbeam}. From \cite{ferrero2015partially} we also know that the eigenvalues of \eqref{plate} may be ordered in an increasing sequence of
strictly positive numbers diverging to $+\infty$. Correspondingly, the eigenfunctions are identified by two indices $m,k\in\N_+$ and they have one of the following forms:
$$W_{m,k}(x,y)=\varphi_{m,k}(y)\sin(mx) \,\quad \text{with corresponding eigenvalue } \nu_{m,k}\, ,$$
$$\overline W_{m,k}(x,y)=\psi_{m,k}(y)\sin(mx)\,\quad \text{with corresponding eigenvalue } \mu_{m,k}\, .$$
The $\varphi_{m,k}$ are odd while the $\psi_{m,k}$ are even and this is why the $W_{m,k}$ are called {\em torsional} eigenfunctions
while the $\overline W_{m,k}$ are called {\em longitudinal} eigenfunctions. The main difference between these two classes is precisely that
$\overline{W}_{m,k}(x,\ell)=\overline{W}_{m,k}(x,-\ell)$ so that the free edges $y=\pm\ell$ are in the same position for longitudinal vibrations, while
$W_{m,k}(x,\ell)=-W_{m,k}(x,-\ell)$ so that the free edges are in opposite positions for torsional vibrations.\par
We first deal with the slightly more complicated case of torsional vibrating modes. Then the eigenvalues $\nu_{m,k}$
are the (ordered) solutions $\lambda > m^4$ of the following equation:
$$\sqrt{\lambda^{1/2} - m^2}[\lambda^{1/2} + (1-\sigma)m^2]^2 \tanh(\ell \sqrt{\lambda^{1/2} + m^2})
= \sqrt{\lambda^{1/2} + m^2}[\lambda^{1/2} - (1-\sigma)m^2]^2 \tanh(\ell \sqrt{\lambda^{1/2} - m^2})\, ,
$$
while the function $\varphi_{m,k}$ may be taken as
$$
\varphi_{m,k}(y)=\big[\nu_{m,k}^{1/2}-\!(1\!-\!\sigma)m^2\big]\tfrac{\sinh\Big(y\sqrt{\nu_{m,k}^{1/2}+m^2}\Big)}{\sinh\Big(\ell\sqrt{\nu_{m,k}^{1/2}+m^2}\Big)}
+\big[\nu_{m,k}^{1/2}+\!(1\!-\!\sigma)m^2\big]\tfrac{\sin\Big(y\sqrt{\nu_{m,k}^{1/2}-m^2}\Big)}{\sin\Big(\ell\sqrt{\nu_{m,k}^{1/2}-m^2}\Big)}\, ,
$$
see \cite{ferrero2015partially}. In particular, $\varphi_{m,k}(\ell) = 2 \sqrt{\nu_{m,k}}=-\varphi_{m,k}(-\ell)$.\par
We view both the free edges of the plate $y=\pm\ell$ as beams connected to a cable and governed by the modified Melan equation \eqref{improved}.
Then we take the following function as a solution of \eqref{improved}:
\begin{equation}\label{eigenvalue1}
w_{m,k}(x):=\alpha W_{m,k}(x,\ell)= \alpha\varphi_{m,k}(\ell)\sin(mx)=2\alpha\sqrt{\nu_{m,k}}\sin(mx)\quad\forall x\in[0,\pi],
\end{equation}
for $m,k\in \mathbb{N}$ and $\alpha\in \mathbb{R}$, a function that belongs to the family of eigenfunctions of \eqref{eigenbeam}, see
\eqref{melan}, assuming that $L=\pi$. As already mentioned, together with $w_{m,k}$ in \eqref{eigenvalue1}, for torsional modes one needs to consider
also its companion $-w_{m,k}$.\par
For longitudinal modes, one has to replace $w_{m,k}$ in \eqref{eigenvalue1} with
\begin{equation}\label{eigenvalue2}
\overline{w}_{m,k}(x):=\alpha\overline{W}_{m,k}(x,\ell)=\alpha\psi_{m,k}(\ell)\sin(mx)\quad\forall x\in[0,\pi],
\end{equation}
where $\psi_{m,k}(\ell)$ depends on the longitudinal eigenvalue $\mu_{m,k}$ of \eqref{plate}; see \cite{ferrero2015partially} for the precise
characterization of $\mu_{m,k}$. For longitudinal modes, the behavior is the same on the two opposite edges.\par
The above discussion, combined with Theorem \ref{theoshortening}, yields the following statement.

\begin{proposition}\label{shortplate}
Assume that $\tfrac{q}{H}<\tfrac{2}{5}$.\par\noindent
$\bullet$ If $m\geq 1$ is even, then the vibrating mode (either torsional or longitudinal) cannot shorten the cable.\par\noindent
$\bullet$ If $m\geq 1$ is odd and the mode is longitudinal, then there exists a (unique) critical value $\alpha^*=\alpha_{m,k}^*>0$
such that for $\alpha\in(0,\alpha^*)$ both the cables are shortened while for other values of $\alpha$ no cable is shortened.\par\noindent
$\bullet$ If $m\geq 1$ is odd and the mode is torsional, then there exists a (unique) critical value $\alpha^*=\alpha_{m,k}^*>0$ such that for
$0<|\alpha|<\alpha^*$ one and only one cable is shortened, while for other values of $\alpha$ no cable is shortened.
\end{proposition}

Following the guideline of Section \ref{thresholdcable}, one may then determine the exact critical values $\alpha_{m,k}^*$ (for odd $m$).
It suffices to consider the critical values $\rho_n^*$ from Theorem \ref{theoshortening} and to take
$$\alpha_{m,k}^*=\frac{\rho_m^*}{\varphi_{m,k}(\ell)}\quad\mbox{or}\quad\alpha_{m,k}^*=\frac{\rho_m^*}{\psi_{m,k}(\ell)}\, ,$$
depending on whether the vibration is torsional or longitudinal.\par
Regarding slackening and the loss of convexity, the above discussion, combined with Propositions \ref{slackening1} and \ref{propmodes},
yields the following statement.

\begin{proposition}\label{slackeningplate}
Let $w$ be the solution of \eqref{improved} and assume that one of the free edges of $\Omega$ is in position $w$. Let $z$ be as in \eqref{z}.
If $\mathcal{S}\neq\emptyset$, then
$$\{x\in(0,L);\, z''(x)\le0\}\subsetneqq\mathcal{S}\, .$$

In particular, if $w_{m,k}$ in \eqref{eigenvalue1} (resp.\ $\overline{w}_{m,k}$ in \eqref{eigenvalue2}) is the position of one of the free edges
of $\Omega$, then slackening of the hangers on that edge occurs if and only if
\begin{eqnarray*}
\alpha_1>\frac{q}{H\varphi_{1,k}(\ell)}\mbox{ when }m=1\, ,&\ &|\alpha_m|>\frac{q}{Hm^2\varphi_{m,k}(\ell)}\mbox{ when }m\ge2\\
\Big(\mbox{resp. }\alpha_1>\frac{q}{H\psi_{1,k}(\ell)}\mbox{ when }m=1\, ,&\ &|\alpha_m|>\frac{q}{Hm^2\psi_{m,k}(\ell)}\mbox{ when }m\ge2\Big)\, ;
\end{eqnarray*}
in this case, the position of the cable is described by $z_{m,k}^{**}$ (with $z_{m,k}=g+w_{m,k}$, resp.\ $z_{m,k}=g+\overline{w}_{m,k}$).
\end{proposition}

The final step consists in considering the evolution equation modeling the vibrations of the partially hinged rectangular plate $\Omega$.
According to \cite{ferrero2015partially}, this leads to the following fourth-order wave-type equation:
\begin{equation} \label{plate_time0}
\left\lbrace
\begin{array}{ll}
u_{tt}+\Delta^2 u = 0\quad & \text{for }(x,y,t)\in\Omega\times\R_+\\
u=u_{xx}=0\quad & \text{for }(x,y,t)\in\{0,\pi\}\times(-\ell,\ell)\times\R_+\\
u_{yy}+\sigma u_{xx}=u_{yyy}+(2-\sigma)u_{xxy}=0\quad & \text{for }(x,y,t)\in(0,\pi)\times\{-\ell,\ell\}\times\R_+\, .
\end{array}
\right.
\end{equation}

We wish to analyze here the evolution of the cable shortening and of the hanger slackening for the torsional vibrating modes $W_{m,k}$ of \eqref{plate};
as for the stationary case, the behavior of the longitudinal modes $\overline{W}_{m,k}$ is simpler.
Therefore, we associate to \eqref{plate_time0} the following initial conditions
\begin{equation}\label{plate_time1}
u(x,y,0) =B W_{m,k}(x,y),\quad u_{t}(x,y,0)=0\qquad\forall (x,y) \in \Omega,
\end{equation}
for some $B\in\mathbb{R}$. The problem \eqref{plate_time0}-\eqref{plate_time1} may be solved by separating variables and the solution is
\begin{equation} \label{plate_time2}
u_{m,k}(x,y,t) = B \cos\big(\sqrt{\nu_{m,k}}\, t\big)W_{m,k}(x,y)\qquad\forall(x,y,t)\in\Omega\times\R_+\, .
\end{equation}
Again, we view both the free edges of $\Omega$ as beams connected to a cable and governed by the modified Melan equation \eqref{improved}.
Therefore, we consider the restriction to the free edge $y=\ell$ (the case $y=-\ell$ being similar) of the function $u_{m,k}$ in \eqref{plate_time2}:
\begin{equation} \label{plate_time3}
v_{m,k}(x,t):=u_{m,k}(x,\ell,t)= B \cos\big(\sqrt{\nu_{m,k}}\, t\big)\varphi_{m,k}(\ell)\sin(mx)\qquad\forall(x,t)\in(0,\pi)\times\R_+\, ,
\end{equation}
see \eqref{eigenvalue1}. Similarly, for the longitudinal modes, we consider the function
\begin{equation} \label{plate_time5}
\overline{v}_{m,k}(x,t):= B \cos\big(\sqrt{\mu_{m,k}}\, t\big)\psi_{m,k}(\ell)\sin(mx)\qquad\forall(x,t)\in(0,\pi)\times\R_+\, ,
\end{equation}
see \eqref{eigenvalue2}. We are interested in determining the conditions under which the cables shorten their length (in odd vibrating modes) or when the hangers slacken.
Unlike the preceding situations, such conditions will now be observed over a space-time region, because the coefficients representing the amplitude of
the expressions \eqref{plate_time3} and \eqref{plate_time5} are periodic functions in time.\par
Concerning the shortening of the cables, we introduce some notations. Let $\alpha_{m,k}^*>0$ be as in Proposition \ref{shortplate}.
If $m\geq 1$ is odd and the mode is longitudinal, put
$$I_S=\{t\ge0;\, 0<B\cos(\sqrt{\mu_{m,k}}\, t)<\alpha_{m,k}^*\}\, ,\quad I_N=\R_+\setminus I_S\, .$$
If $m\geq 1$ is odd and the mode is torsional, put
$$I^S=\{t\ge0;\, 0<|B\cos(\sqrt{\nu_{m,k}}\, t)|<\alpha_{m,k}^*\}\, ,\quad I^N=\R_+\setminus I^{S}\, .$$
Note that all these sets are nonempty, although $I^N$ may have null measure: this happens if $|B|\le\alpha_{m,k}^*$. Then, from Proposition
\ref{shortplate} we deduce the following statement.

\begin{figure}[H]
	\centering
	\includegraphics[scale=0.6] {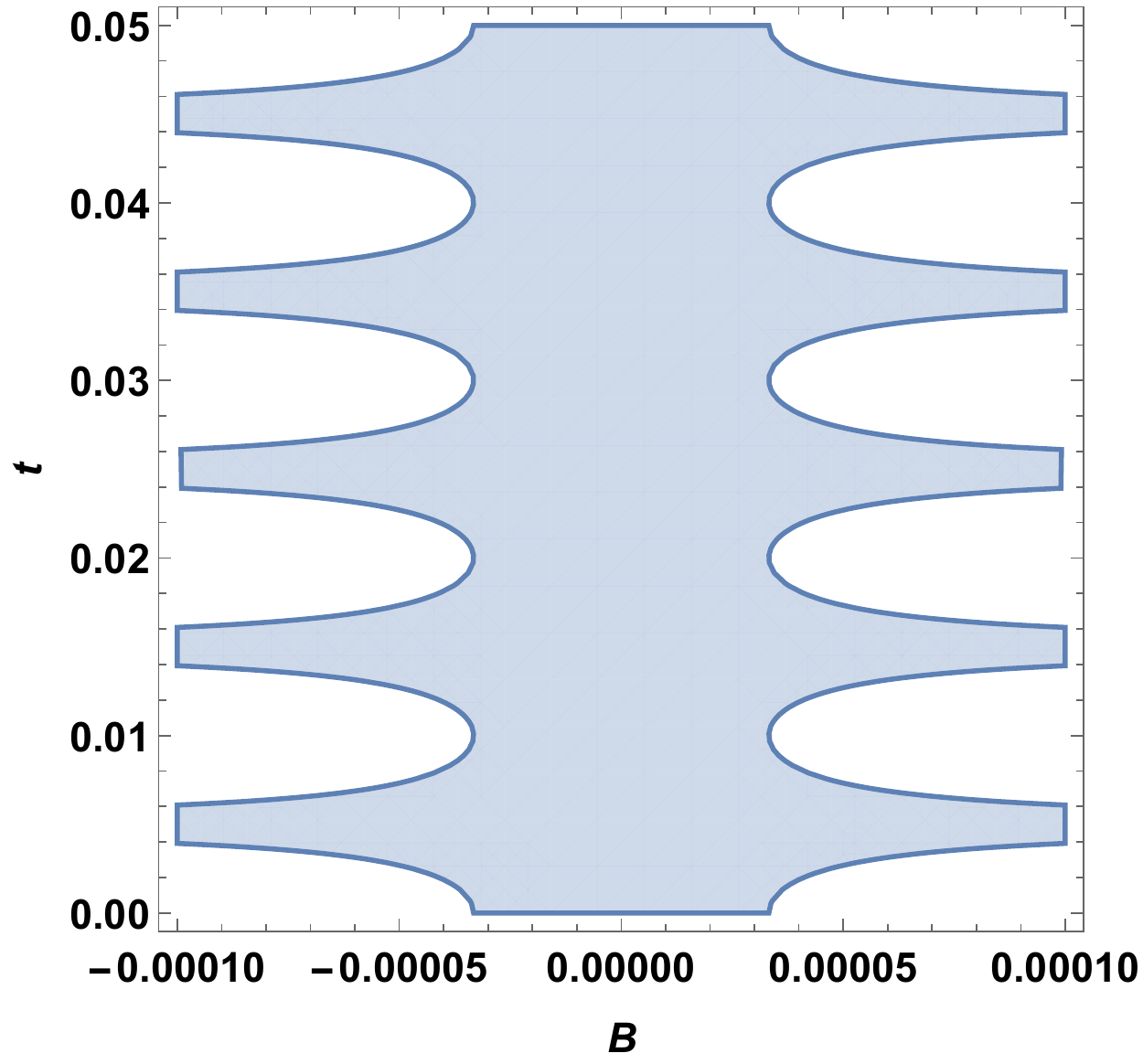}
	\caption{For $B \in [-0.0001,0.0001]$, values of $t\in I^S$ (shaded) provoking cable shortening in the third torsional mode.}\label{subfig:largo}
\end{figure}

\begin{proposition}\label{shortplatevolution}
Assume that $\tfrac{q}{H}<\tfrac{2}{5}$.\par\noindent
$\bullet$ If $m\geq 1$ is even, then the vibrating mode (either torsional or longitudinal) does not shorten the cables for any $t>0$.\par\noindent
$\bullet$ If $m\geq 1$ is odd and the mode is longitudinal, then both the cables are shortened if $t\in I_S$ whereas no cable is shortened if
$t\in I_N$.\par\noindent
$\bullet$ If $m\geq 1$ is odd and the mode is torsional, then one and only one cable is shortened when $t\in I^S$ whereas no cable is shortened if
$t\in I^N$.
\end{proposition}

Once more we emphasize the striking difference between odd and even modes.
Proposition \ref{shortplatevolution} is illustrated in Figure \ref{subfig:largo} by shading the sub-regions of the rectangle
$(B,t)\in[-0.0001 , 0.0001]\times[0,0.05]$ in which $t\in I^S$ for the third torsional mode.
It turns out that for $0<|B|\lessapprox0.000032$, for almost every $t>0$ one (and only one) cable is shortened, whereas for larger values of $|B|$ the white
regions (no shortening) have positive measure.\par
Also the slackening of the hangers (and the loss of convexity) in all the vibrating modes is now observed in a space-time region which
periodically-in-time reproduces itself. In order to discuss together the longitudinal and torsional cases, we use the same notation to denote the
function to be convexified:
\begin{equation} \label{plate_time4}
z_{m,k}(x,t)=g(x)+v_{m,k}(x,t)\ \Big(\mbox{resp. }z_{m,k}(x,t)=g(x)+\overline{v}_{m,k}(x,t)\Big)\qquad\forall(x,t)\in(0,\pi)\times\R_+\, ,
\end{equation}
where $v_{m,k}$ and $\overline{v}_{m,k}$ are as in \eqref{plate_time3} and \eqref{plate_time5}. Concerning the non-convexity regions,
for a given $B\in\mathbb{R}$ they are characterized by the points $(x,t)\in [0,\pi] \times [0,\infty)$ that satisfy the inequality:
$$
\dfrac{\partial^2 z_{m,k}}{\partial x^2}(x,t) \leq 0
$$
or, equivalently, by the points $(x,t) \in [0,\pi] \times [0,\infty)$ in which:
\begin{equation} \label{plate_time6}
B m^2 \varphi_{m,k}(\ell) \cos\left( \sqrt{\nu_{m,k}} \ t \right)\sin \left(m x \right) \geq \frac{q}{H}
\quad\Big(\mbox{resp. }B m^2 \psi_{m,k}(\ell) \cos\left( \sqrt{\nu_{m,k}} \ t \right)\sin \left(m x \right) \geq \frac{q}{H}\Big)\, .
\end{equation}
Notice that inequality \eqref{plate_time6} defines a region of $\mathbb{R}^2$ of positive measure only when $|B|>C_{m,k}^{*}$, where the convexity threshold is
now given by:
$$
C_{m,k}^{*} = \frac{q}{Hm^2\varphi_{m,k}(\ell)} \mbox{ for the torsional modes}, \ \ C_{m,k}^{*} = \frac{q}{Hm^2\psi_{m,k}(\ell)} \mbox{ for the longitudinal modes},
$$
for every integers $m,k\geq 1$. Precisely, given $B \in \mathbb{R}$ and integers $m$ and $k$, let us put:
$$
\alpha_{m,k}(t) = B \cos\left(\sqrt{\nu_{m,k}} \ t \right)\quad\Big(\mbox{resp. }\alpha_{m,k}(t) = B \cos\left(\sqrt{\mu_{m,k}} \ t \right)\Big)\quad\forall t \geq 0.
$$
Then, as a consequence of Proposition \ref{slackeningplate}, we obtain the following statement.

\begin{proposition}\label{slackeningplate2}
Let $u$ be the solution of \eqref{plate_time0} and assume that one of the free edges of $\Omega$ is in position $v_{m,k}$ as in \eqref{plate_time3} or
$\overline{v}_{m,k}$ as in \eqref{plate_time5}. Let $z_{m,k}$ be as in \eqref{plate_time4}, depending on the vibrating mode considered.
If $|B|>C_{m,k}^{*}$, then $\mathcal{S}\neq\emptyset$. Furthermore, whenever $|\alpha_{m,k}(t)| > C_{m,k}^{*}$ we have:
$$
\left\{ x \in (0,\pi) \ | \ \dfrac{\partial^2 z_{m,k}}{\partial x^2}(x,t) \leq 0 \right\} \subsetneqq \mathcal{S}\, .
$$
More precisely, if $|B|>C_{m,k}^{*}$ and if $v_{m,k}$ in \eqref{plate_time3} (resp.\ $\overline{v}_{m,k}$ in \eqref{plate_time5}) is the position of one of the free edges
of $\Omega$, then slackening of the hangers on that edge occurs for all $t>0$ such that:
\begin{eqnarray*}
\alpha_{1,k}(t)>\frac{q}{H\varphi_{1,k}(\ell)}\mbox{ when }m=1\, ,&\ &|\alpha_{m,k}(t)|>\frac{q}{Hm^2\varphi_{m,k}(\ell)}\mbox{ when }m\ge2\\
\Big(\mbox{resp. }\alpha_{1,k}(t)>\frac{q}{H\psi_{1,k}(\ell)}\mbox{ when }m=1\, ,&\ &|\alpha_{m,k}(t)|>\frac{q}{Hm^2\psi_{m,k}(\ell)}\mbox{ when }m\ge2\Big)\, ;
\end{eqnarray*}
in this case, the position of the cable is described by $z_{m,k}^{**}$, with $z_{m,k}=g+v_{m,k}$ as in \eqref{plate_time4}.
\end{proposition}

Proposition \ref{slackeningplate2} defines the slackening regions in the $(x,t)$-plane. Since these are difficult to determine explicitly, we focus
our attention on the non-convexity regions.
As a first example, we take the second torsional mode, whose convexity threshold is $C_{2,1}^{*} \approx 1.52 \times 10^{-4}$.
In this case, setting $B = \pm 2 \times 10^{-4}$ and considering the rectangle $(x,t) \in [0,\pi] \times [0,0.035]$, we obtained Figure \ref{fig:reg2}.
\begin{figure}[H]
	\centering
	\begin{subfigure}{.4\textwidth}
		\includegraphics[scale=0.5] {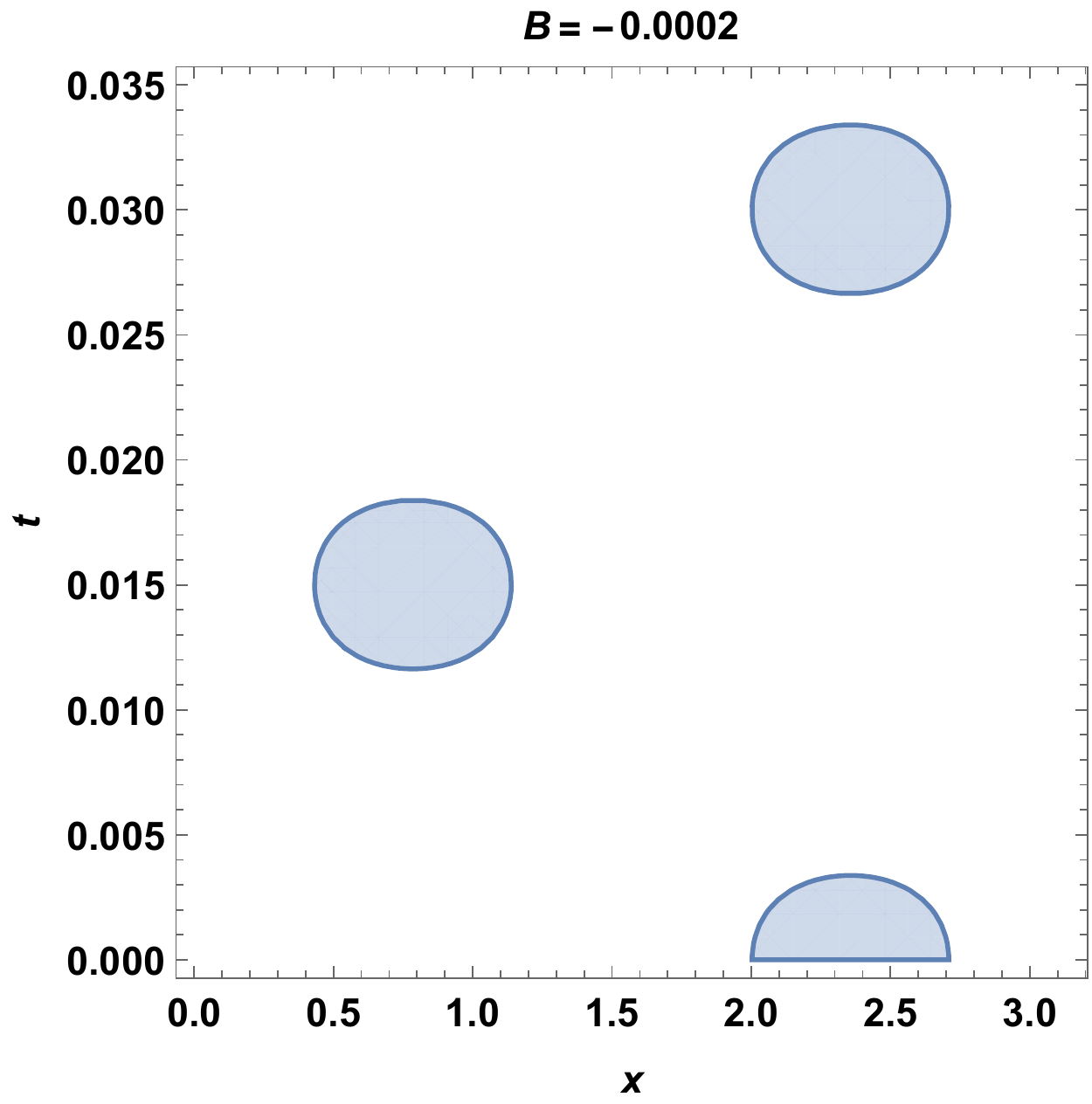}
	\end{subfigure}
	\quad
	\begin{subfigure}{.4\textwidth}
		\includegraphics[scale=0.5]{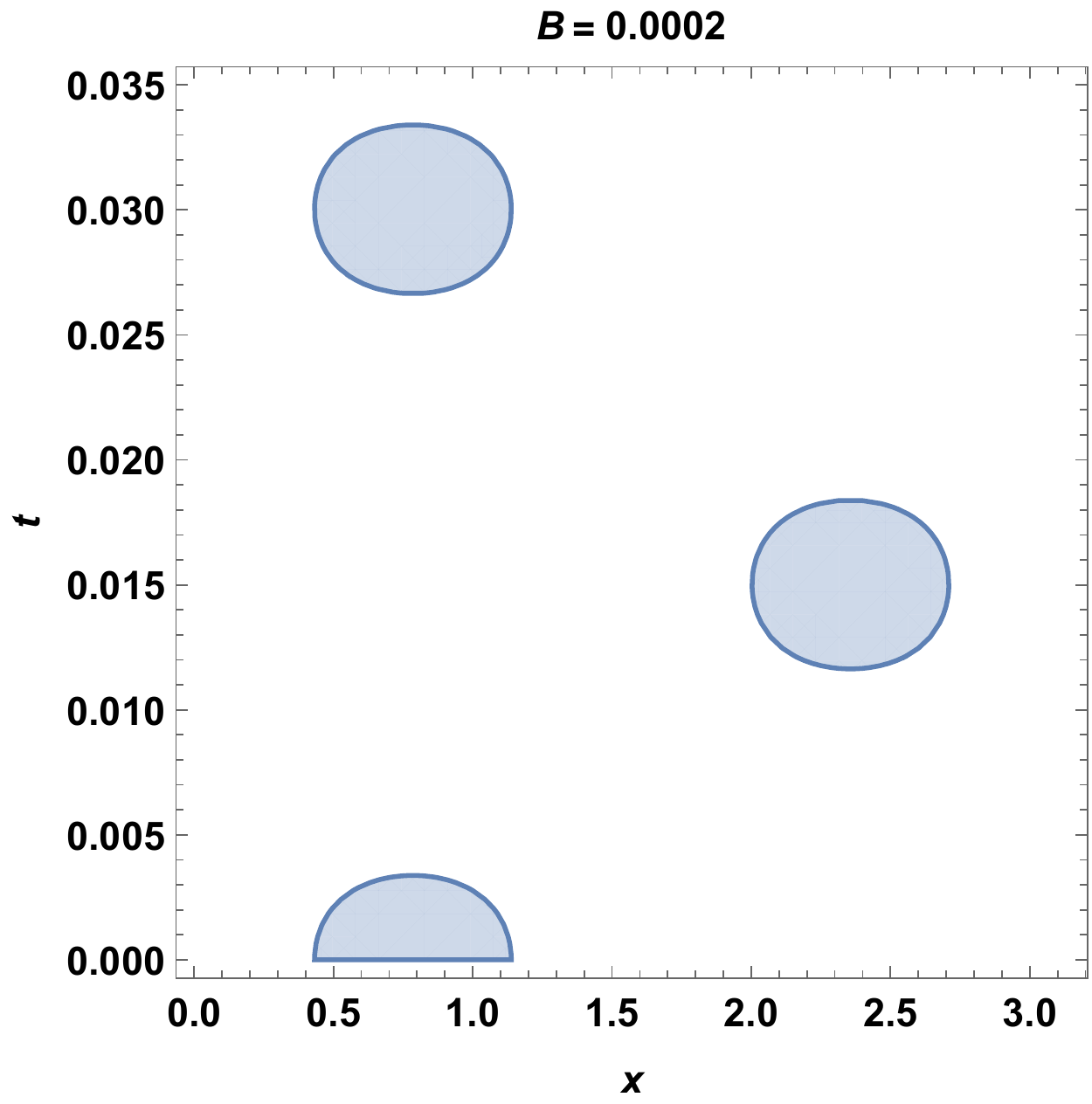}
	\end{subfigure}
    \caption{Non-convexity region of the second torsional mode for a time-varying amplitude.}
    \label{fig:reg2}
\end{figure}

Similar plots are obtained for the function $z_{3,1}$, whose convexity threshold is $C_{3,1}^{*} \approx 4.5 \times 10^{-5}$. By taking $B = \pm 1 \times 10^{-4}$,
we get the following sub-region of the space-time rectangle $(x,t)\in[0,\pi] \times [0,0.025]$ defined by Proposition \ref{slackeningplate2}:
\begin{figure}[H]
	\centering
	\begin{subfigure}{.4\textwidth}
		\includegraphics[scale=0.5] {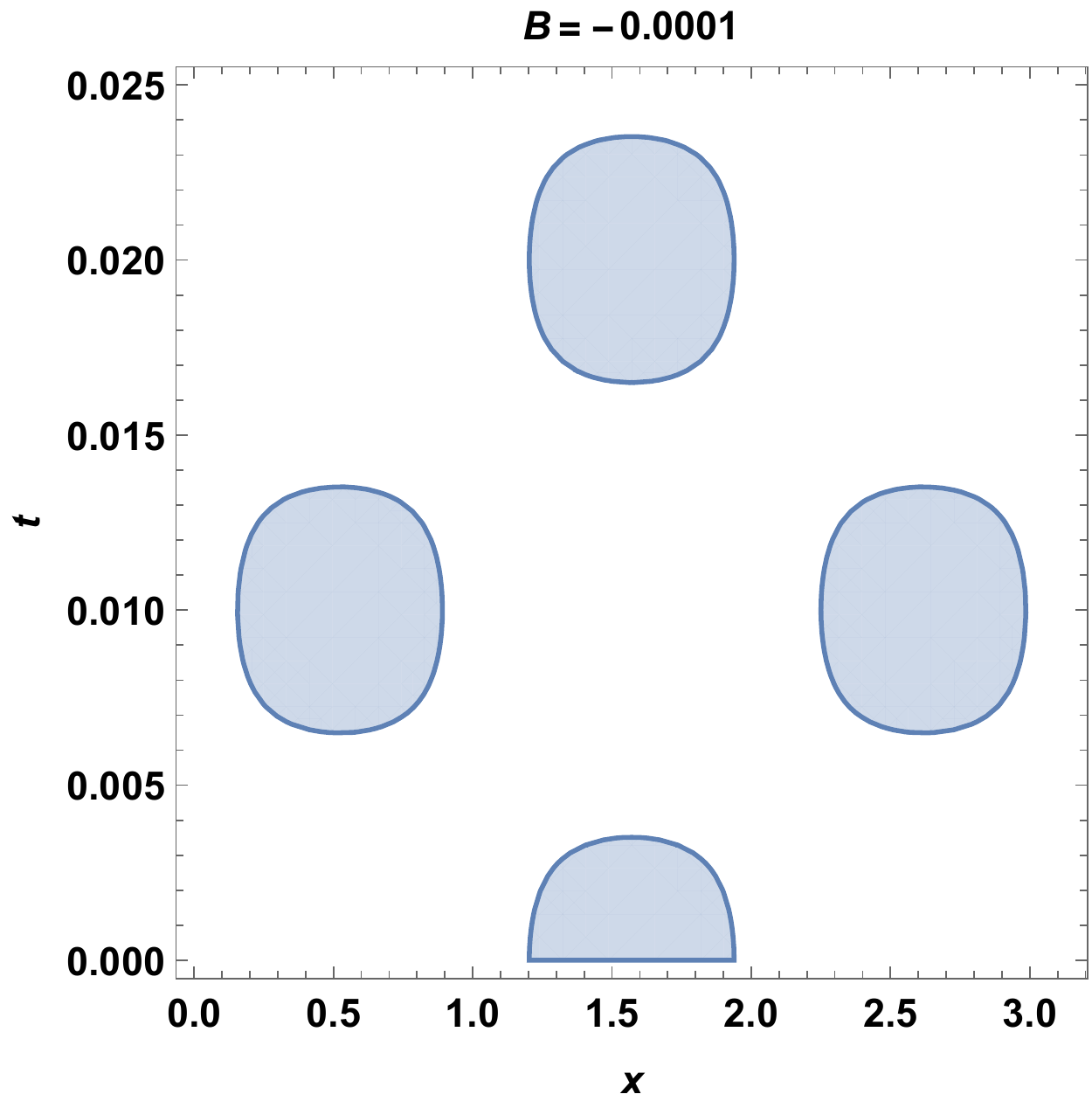}
	\end{subfigure}
	\quad
	\begin{subfigure}{.4\textwidth}
		\includegraphics[scale=0.5]{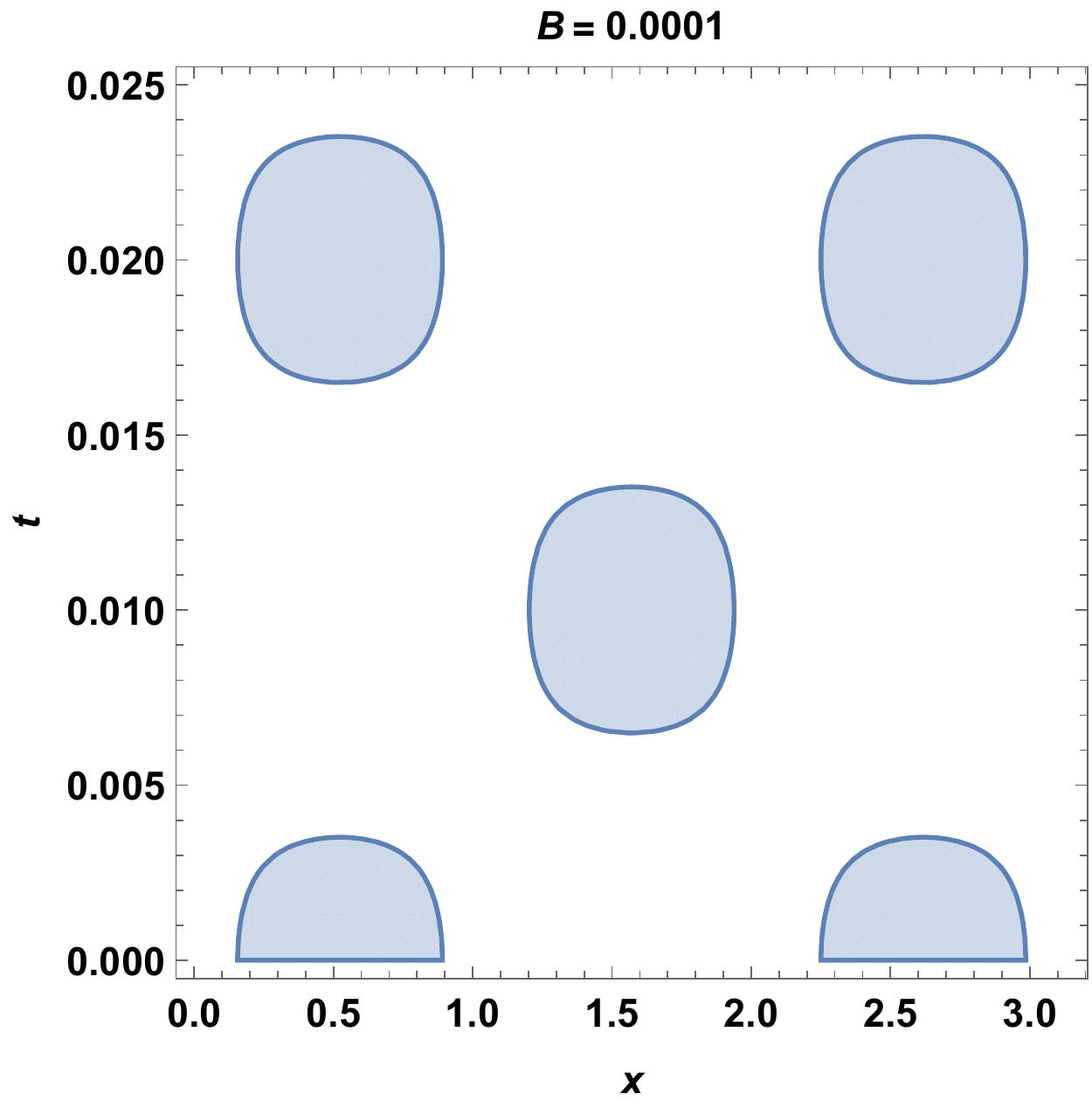}
	\end{subfigure}	
    \caption{Non-convexity region of the third torsional mode for a time-varying amplitude.}
    \label{fig:reg3}
\end{figure}

All these plots may also be read by assuming that the right and left pictures represent simultaneously the non-convexity intervals for each cable, as far as torsional
vibrations are involved: for any given $t>0$ one should cross horizontally the two pictures in order to find which part of the interval $(0,\pi)$ of the two cables
would be non-convex. In fact, the non-convexity regions are proper subsets of the slackening regions, see Proposition \ref{slackeningplate2}. Hence, the slackening
regions are slightly wider in the $x$-direction than the ``ellipses'' in the above plots. This fact is illustrated in Figure \ref{exact} where we compare the
non-convexity and slackening regions in the third torsional mode:

\begin{figure}[H]
\begin{center}
\includegraphics[scale=0.5]{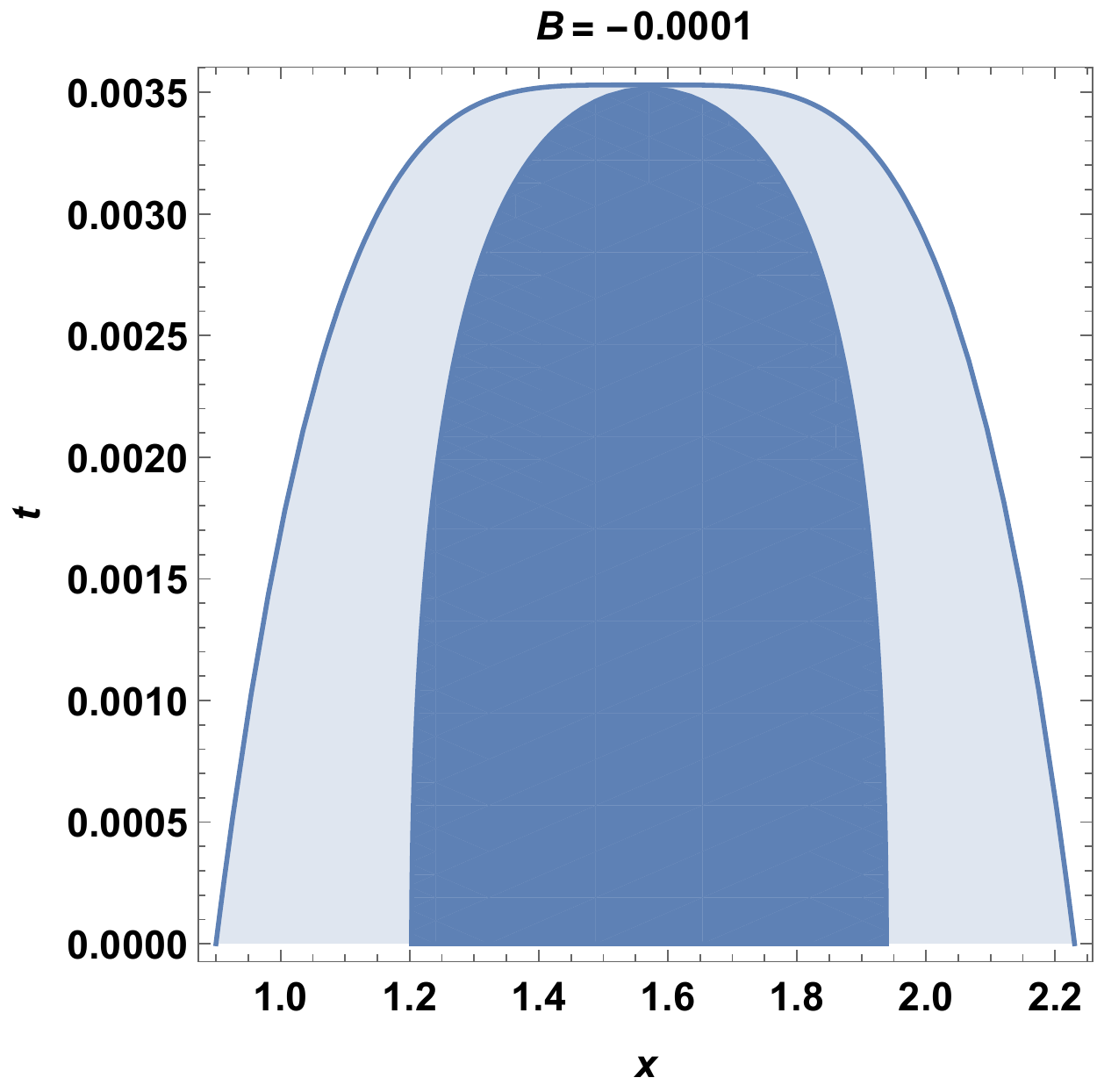}
\caption{Slackening region for the third torsional mode}\label{exact}
\end{center}
\end{figure}

\section{Proof of Theorem \ref{theoshortening}}\label{proof1}

The first step is a technical lemma which involves hyper-geometric integrals:

\begin{lemma} \label{signo_int}
For odd \textnormal{$n \in \mathbb{N}$} and $0<\mu<\frac{2}{5}$ we have
\textnormal{
\begin{equation} \label{integral_pos}
G_{n}:= \int\limits_{0}^{\pi / 2} \dfrac{t \sin(nt)}{\sqrt{1 + \left(\mu t \right)^2}} \, dt\ \left\{\begin{array}{ll}
>0 & \mbox{if }n \equiv 1 \ (\text{mod} \ 4)\\
<0 & \mbox{if }n \equiv 3 \ (\text{mod} \ 4).
\end{array}\right.
\end{equation}}
\end{lemma}
\begin{proof}
For $t \in \mathbb{R}$ such that $| t | < \dfrac{1}{\mu}$, the following power expansion is valid:
$$
\dfrac{1}{\sqrt{1 + \left(\mu t \right)^2}} = \sum_{k=0}^{\infty} \binom{-1 / 2}{k} (\mu t)^{2k}.
$$
Therefore, since $\mu < \frac{2}{5}$, for all $t \in \left[0, \dfrac{\pi}{2} \right]$ we can write:
\begin{equation} \label{lemaint0}
G_{n} = \sum_{k=0}^{\infty} \binom{-1 / 2}{k} I_{n,k} \ \mu^{2k}\quad\forall n \in \mathbb{N},\quad\mbox{where}\quad
I_{n,k} = \int\limits_{0}^{\pi / 2} t^{2k + 1} \sin(nt)\, dt\quad\forall n,k \in \mathbb{N}.
\end{equation}
Since we are considering odd values of $n \in \mathbb{N}$, after integrating by parts twice $I_{n,k}$ in \eqref{lemaint0} we obtain:
$$
I_{n,k} = -\dfrac{2k(2k+1)}{n^2} I_{n,k-1} + \delta(n) \dfrac{2k+1}{n^2} \left( \dfrac{\pi}{2} \right)^{2k}\quad\forall k \geq 1,
$$
with $I_{n,0} = \dfrac{\delta(n)}{n^2}$, for every odd $n \in \mathbb{N}$, and:
\begin{equation} \label{lemaint2}
\delta(n) = \left\{
\begin{aligned}
& 1 \ \ \mbox{if } n \equiv 1 \ (\text{mod} \ 4) \\
- & 1 \ \ \mbox{if } n \equiv 3 \ (\text{mod} \ 4). \\
\end{aligned}
\right.
\end{equation}
An inductive argument over $k \geq 1$ allows then to deduce
\begin{equation} \label{lemaint3}
I_{n,k} = \delta(n) \sum_{j=0}^{k} \dfrac{(-1)^{k + j}}{n^{2(k+1-j)}} \dfrac{(2k+1)!}{(2j)!} \left( \dfrac{\pi}{2} \right)^{2j}\quad\forall k \geq 1.
\end{equation}
Our first claim is that $I_{n,k} > 0$ when \textnormal{$n \equiv 1 \ (\text{mod} \ 4)$}, and that \textnormal{$I_{n,k} < 0$} when
\textnormal{$n \equiv 3 \ (\text{mod} \ 4)$}, for all $k \in \mathbb{N}$. But, according to \eqref{lemaint2} and the form of expression \eqref{lemaint3},
it suffices to show that:
\begin{equation} \label{lemaint4}
J_{n,k} := (-1)^{k} \sum_{j=0}^{k} \dfrac{(-1)^{j}}{(2j)!} \left( \dfrac{n \pi}{2} \right)^{2j} > 0\quad\forall n \in \mathbb{N} \ \text{odd}, \ k \geq 1.
\end{equation}
In order to prove \eqref{lemaint4}, we distinguish two cases.\par\noindent
$\bullet$ \textbf{Case (A):} $k > \dfrac{\sqrt{1 + (n \pi)^2} - 7}{4}$. Since $n \in \mathbb{N}$ is odd, we know that
\begin{equation} \label{lemaint5}
0 = \cos\left( \dfrac{n \pi}{2} \right) = \sum_{j=0}^{k} \dfrac{(-1)^{j}}{(2j)!}
\left( \dfrac{n \pi}{2} \right)^{2j} + \sum_{j=k+1}^{\infty} \dfrac{(-1)^{j}}{(2j)!} \left( \dfrac{n \pi}{2} \right)^{2j}.
\end{equation}
We put $a_{j} = \dfrac{1}{(2j)!} \left( \dfrac{n \pi}{2} \right)^{2j}$ and observe that, for every $j \geq 1$,
$$
\dfrac{a_{j}}{a_{j-1}} = \dfrac{(n\pi)^2}{8j(2j-1)} < 1 \ \ \Longleftrightarrow \ \ j > \dfrac{\sqrt{1 + (n \pi)^2} + 1}{4}.
$$
Hence, the Leibniz criterion can be applied to the \textit{tail} series $\sum_{j=k+1}^\infty(-1)^ja_j$ if $j>\tfrac{\sqrt{1+(n\pi)^2}+1}{4}$.
But since the first ratio to be considered is $a_{k+2}/a_{k+1}$, the Leibniz criterion may be applied whenever
$$
k + 2 > \dfrac{\sqrt{1 + (n \pi)^2} + 1}{4} \ \ \Longleftrightarrow \ \ k > \dfrac{\sqrt{1 + (n \pi)^2} - 7}{4},
$$
which is precisely the case considered. Therefore, the tail series $\sum_{j=k+1}^{\infty} (-1)^{j} a_{j}$ has the same sign as $(-1)^{k + 1}$.
In view of \eqref{lemaint5}, the finite sum $\sum_{j=0}^{k} (-1)^{j} a_{j}$ has the sign of $(-1)^{k}$, that is, the opposite sign of the tail series.
In turn, $J_{n,k} > 0$ in this case, for all odd values of $n \in \mathbb{N}$.\par\noindent
$\bullet$ \textbf{Case (B):} $k < \frac{\sqrt{1 + (n \pi)^2} - 7}{4}$. We distinguish here further between odd and even values of $k$.
For even $k \in \mathbb{N}$, we may write
$$
J_{n,k}= 1 + \sum_{i=1}^{k/2}(a_{2i}-a_{2i-1})
$$
and, since $2i\leq k$, all the terms in the sum are positive in view of the assumption of case B. Therefore, $J_{n,k}>0$ for even $k$.\par
For odd $k\in\mathbb{N}$, we may write
$$
J_{n,k}=\sum_{i=0}^{\frac{k-1}{2}}(a_{2i+1}-a_{2i})
$$
and, since $2i+1\leq k$, all the terms in the sum are positive in view of the assumption of case B. Therefore, $J_{n,k}>0$ also for odd $k$.\par
Inequality \eqref{lemaint4} is so proved for all $n$ and $k$.
Let us now fix an integer $n \equiv 1 \ (\text{mod} \ 4)$ (the case when $n \equiv 3 \ (\text{mod} \ 4)$ follows a completely analogous
procedure). As a consequence of \eqref{lemaint4}, we obtain the upper bound:
\begin{equation} \label{lemaint6}
I_{n,k} = - \dfrac{2k(2k+1)}{n^2} \dfrac{(2k-1)!}{n^{2k}} J_{n,k-1} + \dfrac{2k+1}{n^2} \left( \dfrac{\pi}{2} \right)^{2k} <
\dfrac{2k+1}{n^2} \left( \dfrac{\pi}{2} \right)^{2k}\quad\forall k \geq 1.
\end{equation}
Back to \eqref{lemaint0}, we may write:
\begin{equation} \label{lemaint7}
G_{n} = \dfrac{1}{n^2} + \sum_{k=1}^{\infty} \binom{-1 / 2}{k} \left[\sum_{j=0}^{k} \dfrac{(-1)^{k + j}}{n^{2(k+1-j)}}
\dfrac{(2k+1)!}{(2j)!} \left( \dfrac{\pi}{2} \right)^{2j} \right] \mu^{2k}.
\end{equation}
We observe that the binomial coefficient $\binom{-1 / 2}{k}$ is negative when $k$ is odd and positive otherwise. Furthermore, if we put
$$
b_{k} := \left| \binom{-1 / 2}{k} \right|\quad\forall k \geq 1,
$$
then one has that
\begin{equation} \label{lemaint8}
\dfrac{b_{k+1}}{b_{k}} = \dfrac{2k+1}{2k+2} < 1\quad\forall k \geq 1,
\end{equation}
so that $b_{k} \leq b_{1} = 1/2$ for all $k \geq 1$. Since in \eqref{lemaint7} all the terms in the sum over $j \in \{0,\ldots, k \}$ are strictly positive
as a consequence of \eqref{lemaint4}, by exploiting \eqref{lemaint6} and \eqref{lemaint8} we obtain
$$
G_{n} > \dfrac{1}{n^2} - \sum_{k=1 \atop k\ \textrm{odd}}^{\infty} \dfrac{1}{2} \dfrac{2k+1}{n^2} \left( \dfrac{\mu \pi}{2} \right)^{2k} =
\dfrac{1}{n^2} \left[ 1 - \dfrac{1}{2} \sum_{p=0}^{\infty} (4p+3) \left( \dfrac{\mu \pi}{2} \right)^{4p+2} \right]\, .
$$
For every $x \in (-1,1)$, the geometric series can be differentiated term by term, that is,
$$
\dfrac{d}{dx} \left( \sum_{p=0}^{\infty} x^{4p+3} \right) = \sum_{p=0}^{\infty} (4p + 3) x^{4p+2} =
\dfrac{d}{dx} \left( \dfrac{x^3}{1 - x^4} \right) = \dfrac{x^6 + 3x^2}{(1 - x^4)^2}\qquad\forall x \in (-1,1).
$$
Hence, we finally infer that
\begin{equation} \label{lemaint9}
G_{n} > \dfrac{1}{n^2} \left[ 1 - \dfrac{\left( \dfrac{\mu \pi}{2} \right)^{6} +
3 \left( \dfrac{\mu \pi}{2} \right)^{2}}{2 \left[ 1 - \left( \dfrac{\mu \pi}{2} \right)^{4} \right]^2} \right]\, .
\end{equation}
Some computations show that the right-hand side of \eqref{lemaint9} is strictly positive (at least) when $\dfrac{\mu \pi}{2} < 0.65$,
so in particular, when $\mu < 0.4$. This concludes the proof.
\end{proof}

For the sake of illustration, in Table \ref{table:0} we give the numerical approximation of $G_{n}$, for odd values of $n \in \N$ up to $n=19$,
when $\mu=1.739 \times 10^{-3}$ (as in \eqref{numerical}):
\begin{table}[H]
	\centering
	\begin{tabular}{ | c | c | c | c | c | c | c | c | c | c | c | c | c |}
		\hline
		$n$ & 1 & 3 & 5 & 7 & 9 & 11 & 13 & 15 & 17 & 19 \\ \hline
		$G_{n}$ & 0.9999 & -0.1111 & 0.0399 & -0.0204 & 0.0123 & -0.0082 & 0.0059 & -0.0044  & 0.0034  & -0.0027 \\ \hline
	\end{tabular}
	\captionsetup{justification=centering}
	\captionsetup{font={footnotesize}}
	\captionof{table}{Numerical values of the integral $G_{n}$ in \eqref{integral_pos}, for some odd values of $n \in \N$.}
	\label{table:0}
\end{table}

In fact, for every $\mu \geq 0$ we know that $G_{n} \longrightarrow 0$ as $n \longrightarrow \infty$, as a direct consequence of the Riemann-Lebesgue Theorem.
This is quite visible also in Table \ref{table:0}.\par
Our second technical result gives a qualitative property of the graph of $\Gamma_n(\rho)$.

\begin{lemma}\label{prop_gamma}
For all integer $n\ge1$, the map $\rho\mapsto\Gamma_n(\rho)$ is strictly convex.
\end{lemma}	
\begin{proof} It suffices to analyze the case when $L = \pi$, and so:
$$
\Gamma_{n}(\rho) = \int\limits_{0}^{\pi} \sqrt{1 + \left[\dfrac{q}{H}\left(x - \dfrac{\pi}{2} \right) +n \rho \cos(nx)  \right] ^2} \ dx - L_{c}
\quad\forall \rho \in \mathbb{R}, \ \ \forall n \in \N.
$$
After differentiating under the integral sign we obtain the following:
\begin{equation} \label{dergamma2}
\Gamma_{n}'(\rho) = \int\limits_{0}^{\pi} \dfrac{n \cos(nx) \left[\dfrac{q}{H}
\left(x - \dfrac{\pi}{2} \right) +n \rho \cos(nx)  \right]}{\sqrt{1 + \left[\dfrac{q}{H}\left(x - \dfrac{\pi}{2} \right) +n \rho \cos(nx)  \right]^2}}  \ dx,
\end{equation}
$$
\Gamma_{n}''(\rho)=\int\limits_{0}^{\pi}\dfrac{[n\cos(nx)]^2}{\left[1+\left(\dfrac{q}{H}\left(x-\dfrac{\pi}{2}\right)+n\rho\cos(nx)\right)^2\right]^{3/2}}\ dx,
$$
for $\rho \in \mathbb{R}$ and $n \in \N$. Therefore, $\Gamma_{n}''(\rho) > 0$, for every $n \geq 1$ and $\rho \in \mathbb{R}$, so that $\Gamma_{n}$ is a
strictly convex function all over $\mathbb{R}$.\end{proof}

In view of \eqref{dergamma2}, we see that
\begin{equation} \label{dergamma4}
\Gamma_{n}'(0) = \dfrac{nq}{H} \int\limits_{0}^{\pi} \dfrac{\left(x - \dfrac{\pi}{2} \right)
\cos(nx)}{\sqrt{1 + \left(  \dfrac{q}{H} \right)^2  \left(x - \dfrac{\pi}{2} \right)^2}}  \ dx.
\end{equation}

If $n$ is even, then the integrand in \eqref{dergamma4} is skew-symmetric with respect to $x=\pi/2$ and hence
\begin{equation}\label{even0}
\Gamma_{n}'(0)=0\qquad\mbox{for even }n\, .
\end{equation}

If $n$ is odd, then we make the substitution $t = x - \dfrac{\pi}{2}$ and we note that
$$
\cos \left( nt + \dfrac{n\pi}{2} \right) =
\left\{
\begin{aligned}
- &\sin(nt), \ \mbox{if } n \equiv 1 \ (\text{mod} \ 4) \\
&\sin(nt), \ \mbox{if } n \equiv 3 \ (\text{mod} \ 4), \\
\end{aligned}
\right.
$$
for all $n \geq 1$ and $t \in \left[ -\dfrac{\pi}{2}, \dfrac{\pi}{2} \right]$. Therefore, after setting $\mu=\dfrac{q}{H}$, we see that
$\Gamma_{n}'(0)=-2\mu n \delta(n) G_n$ if $n$ is odd. From \eqref{lemaint2} and Lemma \ref{signo_int} we then infer that
\begin{equation}\label{odd0}
\Gamma_{n}'(0)<0\qquad\mbox{for odd }n\, .
\end{equation}

Since $\Gamma_n(0)=0$ for all $n$, Theorem \ref{theoshortening} follows by combining Lemma \ref{prop_gamma} with \eqref{even0} and \eqref{odd0}.
\par\bigskip\noindent
\textbf{Acknowledgments.} The first author is partially supported by the PRIN project {\em Partial differential equations and related analytic-geometric
inequalities} and by GNAMPA-INdAM.

\bibliographystyle{abbrv}
\bibliography{references}
\end{document}